\documentclass[11pt]{article}
\usepackage{yaonan}

\newcommand{\FPA}{{\sf FPA}}
\newcommand{\OPT}{{\sf OPT}}
\newcommand{\PoA}{{\sf PoA}}
\newcommand{\PoS}{{\sf PoS}}
\newcommand{\SocialWelfare}{\textsf{Social Welfare}}
\newcommand{\SocialWelfares}{\textsf{Social Welfares}}
\newcommand{\FirstPriceAuction}{\textsf{First Price Auction}}
\newcommand{\FirstPriceAuctions}{\textsf{First Price Auctions}}
\newcommand{\PriceofAnarchy}{\textsf{Price of Anarchy}}
\newcommand{\PriceofStability}{\textsf{Price of Stability}}

\newcommand{\BNE}{{\sf BNE}}
\newcommand{\BCE}{{\sf BCE}}
\newcommand{\BCCE}{{\sf BCCE}}
\newcommand{\BayesNashEquilibrium}{\textsf{Bayesian Nash Equilibrium}}
\newcommand{\BayesNashEquilibria}{\textsf{Bayesian Nash Equilibria}}
\newcommand{\BayesCorrelatedEquilibrium}{\textsf{Bayesian Correlated Equilibrium}}

\newcommand{\BayesCoarseCorrelatedEquilibrium}{\textsf{Bayesian Coarse Correlated Equilibrium}}

\newcommand{\alloc}{\mathrm{x}}         
\newcommand{\pay}{\rho}                 
\newcommand{\pays}{\boldsymbol{\rho}}   
\newcommand{\bbFPA}{\mathbb{FPA}}       
\newcommand{\bbBNE}{\mathbb{BNE}}       
\newcommand{\bbBCE}{\mathbb{BCE}}       
\newcommand{\bbBCCE}{\mathbb{BCCE}}     
\newcommand{\bbV}{\mathbb{V}}           

\title{The Price of Stability for First Price Auction}

\author{
Yaonan Jin\thanks{Columbia University. Email: {\tt yj2552@columbia.edu}}
\and
Pinyan Lu\thanks{Shanghai University of Finance and Economics. Email: {\tt lu.pinyan@mail.shufe.edu.cn}}
}
\date{}

\begin{document}

\maketitle
\begin{abstract}
This paper establishes the {\PriceofStability} (\PoS) for {\FirstPriceAuctions}, for all equilibrium concepts that have been studied in the literature: {\BayesNashEquilibrium} $\subsetneq$ {\BayesCorrelatedEquilibrium} $\subsetneq$ {\BayesCoarseCorrelatedEquilibrium}.
\begin{itemize}
    \item {\BayesNashEquilibrium}:
    For independent valuations, the tight $\PoS$ is $1 - 1/ e^{2} \approx 0.8647$, matching the counterpart {\PriceofAnarchy} ({\PoA}) bound \cite{JL22}. 
    For correlated valuations, the tight $\PoS$ is $1 - 1 / e \approx 0.6321$, matching the counterpart {\PoA} bound \cite{ST13,S14}.
    
    This result indicates that, in the worst cases, efficiency degradation depends not on different selections among {\BayesNashEquilibria}.
    
    \item {\sf Bayesian} ({\sf Coarse}) {\sf Correlated Equilibrium}:
    For independent or correlated valuations, the tight {\PoS} is always $1 = 100\%$, i.e., no efficiency degradation, different from the counterpart {\PoA} bound $1 - 1 / e \approx 0.6321$ \cite{ST13,S14}.
    
    This result indicates that {\FirstPriceAuctions} can be fully efficient when we allow the more general equilibrium concepts.
\end{itemize}
\end{abstract}
\thispagestyle{empty}

\newpage

\setcounter{page}{1}

\section{Introduction}
\label{sec:intro}

It is well-known in game theory that a multi-agent system might be in suboptimal states due to selfish behavior of the agents. Auctions are an important genre of such systems. In a single-item auction, each bidder $i \in [n]$ independently draws her value from a distribution $v_{i} \sim V_{i}$ but does not know others' values $\bv_{-i} = (v_{k})_{k \neq i}$. Then, each bidder $i$ submits a (possibly random) bid $b_{i}=s_{i}(v_{i})$  based on her value $v_{i}$ and strategy $s_{i}$.
The auction rule determines the winner and how much the bidders need to pay.
Each bidder $i$ has a quasi-linear utility function $u_{i}(v_{i},\, b_{i}) = v_{i} \cdot \alloc_{i}(b_{i}) - \rho_{i}(b_{i})$, where the winning probability $\alloc_{i}(b_{i})$ and the expected payment $\rho_{i}(b_{i})$ are taken over the randomness of other bidders' values and strategies, as well as the inherent randomness of the auction.
Like other game-theoretical systems, we can define the {\em equilibria} of an auction.

\begin{definition}[Equilibria]
\label{def:bne}
A strategy profile $\bs = \{s_{i}\}_{i \in [n]}$ is a {\BayesNashEquilibrium} for an auction $\calA$ when: For each bidder $i \in [n]$ and any possible value $v \in \supp(V_{i})$, the considered strategy $s_{i}(v)$ is optimal, namely $\Ex_{s_{i}} \big[ u_{i}(v,\, s_{i}(v)) \big] \geq u_{i}(v,\, b)$ for any deviation bid $b \geq 0$.
Denote by $\bbBNE(\bV)$ the space of {\BayesNashEquilibria} of an instance $\bV = \{V_{i}\}_{i \in [n]}$.
\end{definition}

Auctions are widely employed to allocate recourse in a competitive environment. Thus {\em efficiency} is a central property of an auction.
Given an auction $\calA$, the social welfare from an instance $\bV$ at a specific equilibrium $\bs$, denoted by $\calA(\bV,\, \bs)$, is the expectation of the winner's value. Ideally, we would like to allocate the item always to the bidder who values it the most; in expectation, this gives the optimal social welfare $\OPT(\bV)$.

For many auctions, the auction social welfare $\calA(\bV,\, \bs)$ in general is strictly below the optimal social welfare $\OPT(\bV)$. To measure the (in)efficiency of an auction $\calA$, we can define its {\PriceofAnarchy} \cite{KP99} as the {\em worst-case ratio} between the two social welfares.

\begin{definition}[{\PriceofAnarchy}]
\label{def:poa}
The {\PriceofAnarchy} of an auction $\calA$ is given by
\[
    \PoA ~\eqdef~ \inf_{\bV} \inf_{\bs \in \mathbb{BNE}(\bV)} \bigg\{\, \frac{\calA(\bV,\, \bs)}{\OPT(\bV)} \bigg \}.
\]
\end{definition}

We consider the (in)efficiency problem for the {\em first-price auction}, one of the most widely-used auctions. The rule of the first-price auction is very simple: The bidder with the highest bid wins and pays her bid. Simple as the rule is, it is well-known that the equilibria can be very complicated. E.g.\ (\cite{V61}), suppose that there are only two bidders, Alice has a $[0,\, 1]$-uniform random value $v_{1}$ and Bob has a $[0,\, 2]$-uniform random value $v_{2}$, then the unique {\BayesNashEquilibrium} takes the form of $s_{1}(v_{1}) = \frac{4}{3v_{1}}\big(1 - \sqrt{1 - \frac{3}{4}v_{1}^{2}}\big)$ and $s_{2}(v_{2}) = \frac{4}{3v_{2}}\big(\sqrt{1 + \frac{3}{4}v_{2}^{2}} - 1\big)$.


For the {\PoA} in the first-price auction, Syrgkanis and Tardos \cite{ST13} obtained the first nontrivial lower bound of $1 - 1 / e \approx 0.6321$.
Later, Hoy, Taggart, and Wang \cite{HTW18} gave an improved lower bound of $\approx 0.7430$. In a recent work by the authors \cite{JL22}, the tight bound of $1 - 1 / e^2 \approx 0.8647$ was finally derived. That is a complete and insightful characterization.
However, there are still a few remaining issues about the efficiency of the first-price auction, which we will discuss and address in this paper.

First, it is well-known that certain instances may have {\em no} equilibrium. For those instances, the tight {\PoA} bound by \cite{JL22} does not imply anything about the efficiency of the first-price auction. Given this, the natural question is, to what extent can we generalize the tight {\PoA} results?


Second, it is also well-known that certain instances may have {\em multiple} or even {\em infinite} equilibria. For those instances, {\PriceofAnarchy} may be too {\em pessimistic} a measure since it concentrates just on the worst-case equilibria.
Especially, the worst-case instance by \cite{JL22} for the tight $\PoA = 1 - 1 / e^{2}$ does have other more efficient or even fully efficient equilibria.
Towards an {\em optimistic} measure of (in)efficiency, we shall consider another widely studied concept called {\PriceofStability} \cite{ADKTWR08}, which is targeted at the {\em best-case} equilibria (instead of the {\em worst-case} equilibria as for {\PoA}).

\begin{definition}[{\PriceofStability}]
\label{def:pos}
The {\PriceofStability} of an auction $\calA$ is given by
\[
    \PoS ~\eqdef~ \inf_{\bV} \sup_{\bs \in \mathbb{BNE}(\bV)} \bigg\{\, \frac{\calA(\bV,\, \bs)}{\OPT(\bV)} \bigg\}.
\]
\end{definition}

\noindent
By definition, the tight {\PoS} must be lower bounded by the tight {\PoA}. Especially, for the first-price auction, we have $1 - 1 / e^{2} \leq \PoS \leq 1$.

Third, there are other modelings of the (in)efficiency problem. I.e., the above canonical setting assumes (bidder-wise) {\em independent} valuations $\bV$ and strategies $\bs$. Instead, one can consider {\em correlated valuations}, which is quite common in real life.
Also, one can consider {\em correlated strategies}, for which the counterpart solution concepts are (i)~{\BayesCorrelatedEquilibrium} and (ii)~{\BayesCoarseCorrelatedEquilibrium};\footnote{It is well-known (see \cite{R15}) that {\BayesNashEquilibrium} is more special than {\BayesCorrelatedEquilibrium}, which then is more special than {\BayesCoarseCorrelatedEquilibrium}.} see \Cref{sec:BCE_BCCE} for the formal definitions.
In total, we have two valuation classes and three equilibrium concepts, thus $2 \times 3 = 6$ meaningful settings. In {\em each} setting, the {\PoA} and the {\PoS} are both of fundamental interest.

The previous literature studies more on {\PoA} and the tight bounds have been obtained in most settings; see \Cref{subsec:related_work} for a detailed review.
In contrast, the tight {\PoS} bounds remain open in all settings. (Maybe this is because, in each setting, the {\PoS} as the solution to a {\em minimax optimization} problem shall be more challenging than the {\PoA} as the solution to a {\em minimization} problem.)
And understanding those {\PoS} bounds is the main focus of our work.



Besides the concrete bounds, it is also interesting to know in which settings the {\PoS} coincides with the {\PoA}. Namely, if they are equal $\PoA = \PoS$, then this bound is a better characterization of the efficiency since it is ``robust'' against different equilibria.

\subsection{Our results}
\label{subsec:result}

In this work, we will address each of the three issues mentioned above.

For the (possible) non-existence of {\BayesNashEquilibria} in the first-price auction, we show that this is just a consequence of ``the underlying {\em tie-breaking rule} of the auction is incompatible with the considered value distribution $\bV$''.
As a remedy, we prove that, for any $\delta>0$, there always exists a $\delta$-approximate {\BayesNashEquilibrium} that
(i) makes {\em any} tie-breaking rule compatible with the considered value distribution $\bV$, and
(ii) the resulting auction social welfare is at least a $1 - 1 / e^{2}$ fraction of the optimal social welfare.
This indicates that the {\PoA} characterization of the first-price auction is robust and universal. See \Cref{sec:tie_break} for more details.

The main result of our work is the tight {\PoS} bounds in all settings, summarized as follows.

\vspace{.1in}
{\centering
\begin{tabular}{|c|>{\centering\arraybackslash}p{4.75cm}|>{\centering\arraybackslash}p{4.75cm}|}
    \hline
    \rule{0pt}{13pt} & Independent Valuations & Correlated Valuations \\ [2pt]
    \hline
    \rule{0pt}{13pt}{\BNE} & $\PoS = 1 - 1 / e^{2}$ [\Cref{thm:BNE_independent}] & $\PoS = 1 - 1 / e$ [\Cref{thm:BNE_correlated}] \\ [2pt]
    \hline
    \rule{0pt}{13pt}{\BCE} & \multicolumn{2}{c|}{\multirow{2}{*}{\rule{0pt}{13pt}$\PoS = 1$ [\Cref{thm:pos_bce}]}} \\ [2pt]
    \cline{1-1}
    \rule{0pt}{13pt}{\BCCE} & \multicolumn{2}{c|}{} \\ [2pt]
    \hline
\end{tabular}
\par}

\vspace{.1in}
\noindent
Interestingly, in the settings of {\BayesNashEquilibrium} for either independent or correlated valuations, the tight {\PoS} bounds coincide with the {\PoA} counterparts (see \Cref{table:PoA_conditional,table:PoA_unconditional}). This would be an easy corollary if the known {\PoA}-worst instances, due to \cite{JL22} and \cite{S14} respectively, each have unique equilibria. Unfortunately, this is not the case for the either instance. Even worse, the either instance has {\em fully efficient} equilibria, so the {\PoS} bound thereof is $1$.

Towards the tight PoS bounds $= 1 - 1 / e^{2}$ or $1 - 1 / e$, we shall modify the (original) {\PoA}-worst instances from \cite{JL22} and \cite{S14}.
For each modified instance, we first show and verify a particular equilibrium, named by the {\em focal} equilibrium $\bs^{*}$, that is adjusted from the {\em worst-case} equilibrium for the original instance.
More importantly, unlike the original instance, the modification eliminates other {\em more efficient} equilibria, left only with the focal equilibrium $\bs^{*}$. Namely, we prove that the focal equilibrium $\bs^{*}$ is the unique {\BayesNashEquilibrium} of the modified instance.
Furthermore, the modification can be small enough in magnitude, such that the modified auction/optimal social welfares are arbitrarily close to the original counterparts. As a combination, we obtain the identity $\PoS = \PoA = 1 - 1 / e^{2}$ or $1 - 1 / e$ in the either setting.
See \Cref{sec:BNE_independent,sec:BNE_correlated} for more details.

That {\PoA} and {\PoS} have the same tight bounds is conceptually important -- Such a $\PoA = \PoS$ tight bound ``truly'' captures the worst-case efficiency of {\BayesNashEquilibria} in the first-price action, despite the variety of equilibria and the selection among equilibria.

For {\BayesCorrelatedEquilibrium} and/or {\BayesCoarseCorrelatedEquilibrium}, we show that there always exist {\em fully efficient} equilibria.
So, in those settings, whether independent or correlated valuations, we always have $\PoS = 1$. (Notice that regarding a more general equilibrium concept, the {\PoS} becomes larger, while the {\PoA} becomes smaller.)
See \Cref{sec:BCE_BCCE} for more details.

\subsection{Related works}
\label{subsec:related_work}

The first-price auction and its efficiency, motivated by its overwhelming prevalence in real business, are centerpiece of modern auction theory. This study dates back to Vickrey's seminal paper \cite{V61} and has cultivated a rich literature \cite[and the references therein]{SZ90,P92,L96,L99,MR00a,MR00b,JSSZ02,MR03,L06,HKMN11,KZ12,CH13}. However, those works are restricted to special cases -- The equilibria in the first-price auction are notoriously complicated; thus in general, classical economic analysis suffers from certain obstacles.
From a computational perspective, there also is evidence for why the equilibria are elusive \cite{CP14,FGHLP21}.

Over the last two decades, works from computer science bring a fresh viewpoint, {\em approximation guarantees} at the worst-/best-case equilibria, thus coining the concepts ``{\PriceofAnarchy}/{\sf Stability}'' \cite{KP99,ADKTWR08}. Regarding the first-price auction, the state-of-the-art results are summarized in \Cref{table:PoA_unconditional,table:PoA_conditional,table:PoS}.
Notably, \Cref{table:PoA_unconditional,table:PoA_conditional} are row-/column-wise decreasing while \Cref{table:PoS} is row-wise decreasing and column-wise increasing, because the three equilibrium concepts form the inclusion {\BayesNashEquilibrium} $\subseteq$ {\BayesCorrelatedEquilibrium} $\subseteq$ {\BayesCoarseCorrelatedEquilibrium}.
It is remarkable that a standard assumption on the bidders' strategies, called {\em no-overbidding}, can change the tight {\PoA} bounds. In contrast, this assumption never changes the tight {\PoS} bounds in all settings. For more detailed discussions, the reader can refer to the survey \cite{RST17}.

\begin{table}[t]
    {\centering
    \begin{tabular}{|c|>{\centering\arraybackslash}p{4.55cm}|>{\centering\arraybackslash}p{4.61cm}|>{\centering\arraybackslash}p{4.61cm}|}
        \hline
        \rule{0pt}{13pt} & Deterministic Valuations & Independent Valuations & Correlated Valuations \\ [2pt]
        \hline
        \rule{0pt}{13pt}{\BNE} & \multirow{2}{*}{\rule{0pt}{13pt}$\mathrm{TB} = 1$ \hfill folklore \& \cite{FLN16}} & $\mathrm{TB} = 1 - 1 / e^{2}$ \hfill \cite{JL22} & \\ [2pt]
        \cline{1-1}\cline{3-3}
        \rule{0pt}{13pt}{\BCE} & & $1 - 1 / e \leq \mathrm{TB} \leq 1 - 1 / e^{2}$ & \\ [2pt]
        \cline{1-3}
        \rule{0pt}{13pt}{\BCCE} & \multicolumn{2}{c}{} & $\mathrm{TB} = 1 - 1 / e$ \cite{ST13,S14} \\ [2pt]
        \hline
    \end{tabular}
    \par}
    \caption{Tight {\PoA} bounds {\bf without} the no-overbidding assumption. Only one setting, {\BayesCorrelatedEquilibrium} for independent valuations, remains unclear -- No progress apart from the implications $\mathrm{TB} \geq 1 - 1 / e$ \cite{ST13} and $\mathrm{TB} \leq 1 - 1 / e^{2}$ \cite{JL22}, has been made.
    \label{table:PoA_unconditional}}
    \vspace{.2in}
    {\centering
    \begin{tabular}{|c|>{\centering\arraybackslash}p{4.55cm}|>{\centering\arraybackslash}p{4.61cm}|>{\centering\arraybackslash}p{4.61cm}|}
        \hline
        \rule{0pt}{13pt} & Deterministic Valuations & Independent Valuations & Correlated Valuations \\ [2pt]
        \hline
        \rule{0pt}{13pt}{\BNE} & \multirow{2}{*}{\rule{0pt}{13pt}$\mathrm{TB} = 1$ \hfill folklore \& \cite{FLN16}} & $\mathrm{TB} = 1 - 1 / e^{2}$ \hfill \cite{JL22} & \\ [2pt]
        \cline{1-1}\cline{3-3}
        \rule{0pt}{13pt}{\BCE} & & $1 - 1 / e \leq \mathrm{TB} \leq 1 - 1 / e^{2}$ & $\mathrm{TB} = 1 - 1 / e$ \cite{ST13,S14} \\ [2pt]
        \cline{1-3}
        \rule{0pt}{13pt}{\BCCE} & $\mathrm{TB} \approx 81.36\%$ \hfill \cite{FLN16} & $1 - 1 / e \leq \mathrm{TB} \lessapprox 81.36\%$ & \\ [2pt]
        \hline
    \end{tabular}
    \par}
    \caption{Tight {\PoA} bounds {\bf with} the no-overbidding assumption. Two settings,
    {\BayesCorrelatedEquilibrium} and {\BayesCoarseCorrelatedEquilibrium} for independent valuations, remain unclear -- No progress apart from the implications $\mathrm{TB} \geq 1 - 1 / e$ for the both settings \cite{ST13}, $\mathrm{TB} \leq 1 - 1 / e^{2}$ for the {\BCE} setting \cite{JL22}, and $\mathrm{TB} \lessapprox 81.36\%$ for the {\BCCE} setting \cite{FLN16}, has been made.
    \label{table:PoA_conditional}}
    \vspace{.2in}
    {\centering
    \begin{tabular}{|c|>{\centering\arraybackslash}p{4.55cm}|>{\centering\arraybackslash}p{4.61cm}|>{\centering\arraybackslash}p{4.61cm}|}
        \hline
        \rule{0pt}{13pt} & Deterministic Valuations & Independent Valuations & Correlated Valuations \\ [2pt]
        \hline
        \rule{0pt}{13pt}{\BNE} & & $\mathrm{TB} = 1 - 1 / e^{2}$ [\Cref{thm:BNE_independent}] & $\mathrm{TB} = 1 - 1 / e$ [\Cref{thm:BNE_correlated}] \\ [2pt]
        \cline{1-1}\cline{3-4}
        \rule{0pt}{13pt}{\BCE} & $\mathrm{TB} = 1$ \qquad folklore & \multicolumn{2}{c|}{\multirow{2}{*}{\rule{0pt}{13pt}$\mathrm{TB} = 1$ [\Cref{thm:pos_bce}]}} \\ [2pt]
        \cline{1-1}
        \rule{0pt}{13pt}{\BCCE} & & \multicolumn{2}{c|}{} \\ [2pt]
        \hline
    \end{tabular}
    \par}
    \caption{Tight {\PoS} bounds {\bf regardless of} the no-overbidding assumption. All settings are clear.
    \label{table:PoS}}
\end{table}

Technically, the most prevalent tool for studying {\PriceofAnarchy} in auctions is the smoothness framework proposed by Roughgarden \cite{R15} and then developed by Syrgkanis and Tardos \cite{ST13}. This framework enables the tight bound $= 1 - 1 / e$ in most settings, but has inherent bottlenecks in the canonical setting, namely {\BayesNashEquilibrium} for independent valuations. To mitigate those issues, Hoy, Taggart, and Wang \cite{HTW18} combined additional techniques into the smoothness framework, hence an improved lower bound of $\approx 0.7430$. Until very recently, through a completely new framework, the authors \cite{JL22} finally derived the tight bound $= 1 - 1 / e^{2} \approx 0.8647$.

The above discussions all concern efficiency guarantees. Another interesting and relevant topic is revenue guarantees in the first-price auction. Hartline, Hoy, and Taggart \cite{HHT14} showed that, when the auctioneer sets {\em bidder-personalized} reserves in the first-price auction, the worst-case equilibria achieve a $\geq \frac{1}{2}(1 - 1 / e) \approx 31.61\%$ approximation to optimal revenues. As an implication of the later works \cite{AHNPY19,JLQTX19a}, a better revenue guarantee $\gtrapprox \frac{1}{2.6202} \approx 38.17\%$ holds even when the auctioneer sets {\em bidder-anonymous} reserves. It would be interesting to capture the \textsf{revenue-PoA} and \textsf{revenue-PoS} for the first-price auction with (optimal) personalized/anonymous reserves.

\section{Notation and Preliminaries}
\label{sec:prelim}

This section presents a bunch of structural results from the literature, especially \cite{JL22}, which lay the foundation of our paper. (More structural results will be presented in the later sections, when they are needed for our discussions.)


In a single-item auction, the bidders $[n] = \{1,\, 2,\, \dots,\, n\}$ submit {\em non-negative} bids $\bb = (b_{i})_{i \in [n]}$ to the auctioneer.
{\FirstPriceAuction} is a family of auctions $\calA = (\alloc,\, \pays)$ that all obey the first-price allocation/payment principles.
\begin{itemize}
    \item {\bf first-price allocation:}
    Let $X(\bb) \eqdef \argmax \{b_{i}: i \in [n]\}$. If there is one unique first-order bidder $|X(\bb)| = 1$, allocate the item to her $\alloc(\bb) \equiv X(\bb)$. Otherwise $|X(\bb)| \geq 2$, allocate the item to one of those first-order bidders $\alloc(\bb) \in X(\bb)$, via some (randomized) {\em tie-breaking} rule for this bid profile $\bb$.

    \item {\bf first-price payment:}
    The allocated bidder $\alloc(\bb)$ pays her own bid, while the non-allocated bidders $[n] \setminus \{\alloc(\bb)\}$ pay nothing. Formally, $\pay_{i}(\bb) = b_{i} \cdot \indicator(i = \alloc(\bb))$ for each $i \in [n]$.
\end{itemize}
Hence, different {\FirstPriceAuctions} $\calA \in \bbFPA$ are identified by their allocation/tie-breaking rules $\alloc(\bb)$ and, without ambiguity, we can abuse the notation $\alloc \in \bbFPA$.

Regarding a joint value distribution $\bv = (v_{i})_{i \in [n]} \sim \bV \in \bbV_{\sf joint}$, a (randomized) strategy profile $\bs = \{s_{i}\}_{i \in [n]}$ maps the realized individual values $v_{i}$ to the (random) individual bids $s_{i}(v_{i})$.
Over the randomness of other bidders' bids $\bs_{-i}(\bv_{-i})$ and the allocation rule $\alloc \in \bbFPA$, bidder $i \in [n]$ on having a value $v \geq 0$ and a bid $b \geq 0$ wins with probability $\alloc_{i}(b) \eqdef \Pr_{\bv,\, \bs,\, \alloc} [ i = \alloc(b,\, \bs_{-i}(\bv_{-i})) \mid v_{i} = v]$ and gains an interim utility $u_{i}(v,\, b) \eqdef (v - b) \cdot \alloc_{i}(b)$. Such a strategy profile $\bs$ forms a {\BayesNashEquilibrium} when it satisfies the following conditions.

\begin{definition}[{\BayesNashEquilibria}]
\label{def:bne_formal}
Given a joint value distribution $\bV \in \bbV_{\sf joint}$, an allocation rule $\alloc \in \bbFPA$, and a precision $\delta > 0$:
\begin{itemize}
    \item An (exact) {\BayesNashEquilibrium} $\bs \in \bbBNE(\bV,\, \alloc)$ is a strategy profile $\bs = \{s_{i}\}_{i \in [n]}$ that, for any bidder $i \in [n]$, any value of her $v \in \supp_{i}(\bV)$, and any deviation bid $b^{*} \geq 0$,
    \begin{align*}
        \Ex_{\bv,\, \bs,\, \alloc} [\, u_{i}(v_{i},\, \bs(\bv)) \,\mid\, v_{i} = v \,]
        ~\geq~ \Ex_{\bv,\, \bs,\, \alloc} [\, u_{i}(v_{i},\, b^{*},\, \bs_{-i}(\bv)) \,\mid\, v_{i} = v \,].
        \phantom{ - \delta}
    \end{align*}
    
    \item A $\delta$-approximate {\BayesNashEquilibrium} $\bs \in \bbBNE(\bV,\, \alloc)$ is a strategy profile $\bs = \{s_{i}\}_{i \in [n]}$ that, for any bidder $i \in [n]$, any value of her $v \in \supp_{i}(\bV)$, and any deviation bid $b^{*} \geq 0$,
    \begin{align*}
        \Ex_{\bv,\, \bs,\, \alloc} [\, u_{i}(v_{i},\, \bs(\bv)) \,\mid\, v_{i} = v \,]
        ~\geq~ \Ex_{\bv,\, \bs,\, \alloc} [\, u_{i}(v_{i},\, b^{*},\, \bs_{-i}(\bv)) \,\mid\, v_{i} = v \,] - \delta.
    \end{align*}
\end{itemize}
\end{definition}

\subsection{Independent valuations}

When the value distribution $\bV$ degenerates into a product value distribution $\bV = \{V_{i}\}_{i \in [n]} \in \bbV_{\sf prod}$, the equilibria thereof have several remarkable properties, which we give a brief review here.

First, the following result on the existence of exact equilibria can be concluded from \cite{L96}.

\begin{proposition}[{\cite{L96}}]
\label{thm:exist_bne}
Given a product value distribution $\bV = \{V_{i}\}_{i \in [n]} \in \bV_{\sf prod}$, there exists some
tie-breaking rule $\alloc \in \bbFPA$ such that the resulting {\FirstPriceAuction} admits at least one exact equilibrium $\bbBNE(\bV,\, \alloc) \neq \emptyset$.
\end{proposition}

\noindent
Given an exact {\BayesNashEquilibrium} $\bs \in \bbBNE(\bV,\, \alloc)$, we will adopt the following notations.
\begin{itemize}
    \item $\bB = \{B_{i}\}_{i \in [n]}$ denotes the equilibrium bid distributions $\bs(\bv) = (s_{i}(v_{i}))_{i \in [n]} \sim \bB$.
    
    \item $\calB(b) = \prod_{i \in [n]} B_{i}(b)$ denotes the first-order bid distribution $\max(\bs(\bv)) \sim \calB$.

    \item $\calB_{-i}(b) = \prod_{k \in [n] \setminus \{i\}} B_{k}(b)$ denotes the competing bid distribution of each bidder $i \in [n]$.

    \item $\gamma \eqdef \inf(\supp(\calB))$ and $\lambda \eqdef \sup(\supp(\calB))$ denote the ``infimum''/``supremum'' first-order bids, respectively.
    Without ambiguity, we call $v,\, b < \gamma$ the {\em \textbf{low}} values/bids, $v,\, b = \gamma$ the {\em \textbf{boundary}} values/bids, and $v,\, b > \gamma$ the {\em \textbf{normal}} values/bids. In other words:
    (i)~low bids $b < \gamma$ give a zero winning probability and are less important;
    (ii)~normal bids $b > \gamma$ are the most common bids and will behave nicely; and
    (iii)~boundary bids $b = \gamma$ are tricky and will be dealt with separately.
\end{itemize}
The next proposition, due to \cite[Lemma~2.7]{JL22}, shows that the equilibrium/competing/first-order bid distributions $B_{i}(b)$, $\calB_{-i}(b)$, and $\calB(b)$ have nice structures.

\begin{proposition}[{\cite[Lemma~2.7]{JL22}}]
\label{lem:bid_distribution}
Each of the following holds:
\begin{enumerate}[font = {\em\bfseries}]
    \item\label{lem:bid_distribution:monotonicity}
    {\bf monotonicity:} The competing/first-order bid distributions $\{\calB_{-i}\}_{i \in [n]}$ and $\calB$ each have probability densities almost everywhere on $b \in (\gamma,\, \lambda]$, thus having strictly increasing CDF's on the closed interval $b \in [\gamma,\, \lambda]$.

    \item\label{lem:bid_distribution:continuity}
    {\bf continuity:} The equilibrium/competing/first-order bid distributions $\{B_{i}\}_{i \in [n]}$, $\{\calB_{-i}\}_{i \in [n]}$ and $\calB$ each have no probability mass on $b \in (\gamma,\, \lambda]$, excluding the boundary $\gamma = \inf(\supp(\calB))$, thus having continuous CDF's on the closed interval $b \in [\gamma,\, \lambda]$.
\end{enumerate}
\end{proposition}

Two more requisite notions for our later discussions are {\em bid-to-value mappings} and {\em monopolists} (\Cref{def:mapping,def:monopolist}). Particularly, we will leverage two structural results also from \cite{JL22}.

\begin{definition}[Bid-to-value mappings]
\label{def:mapping}
\begin{flushleft}
The bid-to-value mappings $\bvarphi = \{\varphi_{i}\}_{i \in [n]}$ are defined as $\varphi_{i}(b) \eqdef b + \calB_{-i}(b) / \calB'_{-i}(b) = b + (\sum_{k \in [n] \setminus \{i\}} B'_{k}(b) / B_{k}(b))^{-1}$ for $b \in (\gamma,\, \lambda)$.
\end{flushleft}
\end{definition}

\begin{definition}[Monopolists]
\label{def:monopolist}
\begin{flushleft}
A bidder $h \in [n]$ is called a {\em monopolist} when the probability of taking a normal value yet a boundary bid is nonzero $\Pr_{v_{h},\, s_{h}} [(v_{h} > \gamma) \wedge (s_{h}(v_{h}) = \gamma)] > 0$.
\end{flushleft}
\end{definition}

\begin{proposition}[{\cite[Lemma~2.13]{JL22}}]
\label{lem:high_bid}
Each bid-to-value mapping $\varphi_{i}(b)$ for $i \in [n]$ is increasing on the open interval $b \in (\gamma,\, \lambda)$. Therefore, the domain can be extended to include the both endpoints $\varphi_{i}(\gamma) \eqdef \lim_{b \searrow \gamma} \varphi_{i}(b)$ and $\varphi_{i}(\lambda) \eqdef \lim_{b \nearrow \lambda} \varphi_{i}(b)$.
\end{proposition}

\begin{proposition}[{\cite[Lemma~2.16]{JL22}}]
\label{lem:monopolist}
There exists at most one monopolist $h \in [n]$. If existential: \\
{\bf (I)}~A boundary first-order bid $\{ \max(\bb) = \gamma \}$ occurs with a nonzero probability $\calB(\gamma) > 0$. \\
{\bf (II)}~Conditioned on the tiebreak $\{ b_{h} = \max(\bb) = \gamma \}$, the monopolist wins $\alloc(\bb) = h$ almost surely.
\end{proposition}

\newcommand{\EMD}{{\sf EMD}}

\section{Tie-breaking Rules and Approximate Equilibria}
\label{sec:tie_break}


In this section, we discuss the existence of exact/approximate equilibria for a product value distribution $\bV = \{V_{i}\}_{i \in [n]} \in \bbV_{\sf prod}$. Following \Cref{thm:exist_bne}, the only possibility for nonexistence of exact equilibria is that the underlying tie-breaking rule $\alloc \in \bbFPA$ may be incompatible with this value distribution. Instead, we will start with a {\em compatible} tie-breaking and an {\em exact} equilibrium thereof $\bs \in \bbBNE(\bV,\, \alloc)$.
Then for any given $\delta>0$, we slightly modify this equilibrium into a new strategy profile $\bs^*$ that is insensitive to different tie-breaking rules. I.e., this strategy profile $\bs^{*}$ is a {\em universal} $\delta$-approximate equilibrium for {\FirstPriceAuction}, regardless of the tie-breaking rules.

To make the modification workable, we crucially leverage several structural results from \cite{JL22} about {\BayesNashEquilibrium}. In particular, we will use the fact that nontrivial tie-breaks can occur only when the first-order bid is at the boundary $\gamma = \inf(\supp(\calB))$.

Before giving the formal statement of our result, we recall the concept of {\em earth mover's distance} \cite[Chapter~6]{V09}, which will be used
to measure the distance between two strategies.

\begin{definition}[Earth Mover's Distance]
Given two single-dimensional distributions $D$ and $\tilde{D}$, denote by $D^{-1}(q)$ and $\tilde{D}^{-1}(q)$ for $q \in [0,\, 1]$ the quantile functions, then:
\begin{itemize}
    \item The $\ell_{p}$-norm earth mover's distance, for $p \geq 1$, is defined as
    \[
        \EMD_{p}(D,\, \tilde{D}) ~=~ \Big(\int_{0}^{1} \big| D^{-1}(q) - \tilde{D}^{-1}(q) \big|^{p} \cdot \d q\Big)^{1 / p}.
    \]
    
    \item The $\ell_{\infty}$-norm earth mover's distance is defined as
    \[
        \EMD_{\infty}(D,\, \tilde{D}) ~=~ \sup \Big\{\big|D^{-1}(q) - \tilde{D}^{-1}(q)\big|: q \in [0,\, 1]\Big\}.
    \]
\end{itemize}
It follows that $\EMD_{p_{1}}(D,\, \tilde{D}) \leq \EMD_{p_{2}}(D,\, \tilde{D}) \leq \EMD_{\infty}(D,\, \tilde{D})$ for any $p_{2} \geq p_{1} \geq 1$.
\end{definition}

Below, \Cref{thm:apx_bne} summarizes our result on the existence of universal $\delta$-approximate equilibria. The proof relies on \cite[Lemma~2.5]{JL22}.

\begin{lemma}[Bidding Dichotomy {\cite{JL22}}]
\label{lem:dichotomy}
At an exact {\BayesNashEquilibrium} $\bs \in \bbBNE(\bV,\, \alloc)$, for each bidder $i \in [n]$, the following hold almost surely:
\begin{flushleft}
\begin{enumerate}[font = {\em\bfseries}]
    \item\label{lem:dichotomy:1}
    A low/boundary value $v \in \supp_{\leq \gamma}(V_{i})$ induces a low/boundary equilibrium bid $s_{i}(v) \leq \gamma$.

    \item\label{lem:dichotomy:2}
    A normal value $v \in \supp_{> \gamma}(V_{i})$ induces a boundary/normal equilibrium bid $\gamma \leq s_{i}(v) < v$.
\end{enumerate}
\end{flushleft}
\end{lemma}

\begin{theorem}[{\BayesNashEquilibria}]
\label{thm:apx_bne}
Given a product value distribution $\bV = \{V_{i}\}_{i \in [n]} \in \bbV_{\sf prod}$, a tie-breaking rule $\alloc \in \bbFPA$, and an exact {\BayesNashEquilibrium} thereof $\bs \in \bbBNE(\bV,\, \alloc) \neq \emptyset$.
For any precision $\delta > 0$, there exists another strategy profile $\bs^{*} = \{s_{i}^{*}\}_{i \in [n]}$ such that:
\begin{enumerate}[font = {\em\bfseries}]
    \item\label{thm:apx_bne:closeness}
    {\bf closeness:}
    $\EMD_{\infty}(s_{i}(v),\, s_{i}^{*}(v)) \leq \delta$ for any value $v \in \supp(V_{i})$ and each bidder $i \in [n]$.\footnote{Recall that each strategy $s_{i}$ or $s_{i}^{*}$ is a family of bid distributions indexed by the value $v \in \supp(V_{i})$.}
    
    \item\label{thm:apx_bne:efficiency}
    {\bf efficiency invariant:}
    For an arbitrary tie-breaking rule $\alloc^{*} \in \bbFPA$ (possibly the same as $\alloc$), the expected optimal/auction {\SocialWelfares} keep the same $\OPT(\bV,\, \alloc^{*},\, \bs^{*}) = \OPT(\bV,\, \alloc,\, \bs)$ and $\FPA(\bV,\, \alloc^{*},\, \bs^{*}) = \FPA(\bV,\, \alloc,\, \bs)$.
    
    \item\label{thm:apx_bne:universality}
    {\bf universality:}
    For an arbitrary tie-breaking rule $\alloc^{*} \in \bbFPA$ (possibly the same as $\alloc$), it forms a $\delta$-approximate equilibrium $\bs^{*} \in \bbBNE(\bV,\, \alloc^{*},\, \delta)$.
    Formally, it forms a universal $\delta$-approximate equilibrium $\bs^{*} \in \big(\bigcap_{\alloc^{*} \in \bbFPA} \bbBNE(\bV,\, \alloc^{*},\, \delta)\big)$.
\end{enumerate}
\end{theorem}

\begin{proof}
There are $9$ kinds of tuples $(v_{i},\, s_{i}(v_{i}))$, i.e., low/boundary/normal values $v_{i}$ and bids $s_{i}(v_{i})$. Based on case analysis, we construct the new strategy profile $\bs^{*} = \{s_{i}^{*}\}_{i \in [n]}$ in a {\em coupling} way.
\begin{center}
    \begin{tabular}{|l||*{3}{l|}}\hline
        \backslashbox{value}{bid} & low $s_{i}(v_{i}) < \gamma$ & BDY $s_{i}(v_{i}) = \gamma$ & normal $s_{i}(v_{i}) > \gamma$ \\
        \hline
        \hline
        \rule{0pt}{15pt}low $v_{i} < \gamma$ & $s_{i}^{*}(v_{i}) = s_{i}(v_{i}) - \delta$ & $s_{i}^{*}(v_{i}) = s_{i}(v_{i}) - \delta$ & impossible (\Cref{lem:dichotomy}) \\ [4pt]
        \hline
        \rule{0pt}{15pt}BDY $v_{i} = \gamma$ & $s_{i}^{*}(v_{i}) = s_{i}(v_{i}) - \delta$ & $s_{i}^{*}(v_{i}) = s_{i}(v_{i}) - \delta / 2$ & impossible (\Cref{lem:dichotomy}) \\ [4pt]
        \hline
        \rule{0pt}{15pt}normal $v_{i} > \gamma$ & impossible (\Cref{lem:dichotomy}) & $s_{i}^{*}(v_{i}) = s_{i}(v_{i})$ & $s_{i}^{*}(v_{i}) = s_{i}(v_{i})$ \\ [4pt]
        \hline
    \end{tabular}
\end{center}

\noindent
This coupling intrinsically ensures {\bf \Cref{thm:apx_bne:closeness}} that $\EMD_{\infty}(s_{i}(v),\, s_{i}^{*}(v)) \leq \delta$ for any value $v \in \supp(V_{i})$ and each bidder $i \in [n]$. It remains to show {\bf \Cref{thm:apx_bne:efficiency}} and {\bf \Cref{thm:apx_bne:universality}}.

\vspace{.1in}
\noindent
{\bf \Cref{thm:apx_bne:efficiency}.}
The expected optimal {\SocialWelfare}, which relies just on the value distribution $\bV$, must be invariant $\OPT(\bV,\, \alloc^{*},\, \bs^{*}) = \OPT(\bV,\, \alloc,\, \bs)$. To reason about the expected auction {\SocialWelfare}, recall \Cref{lem:monopolist} that the allocated bidder $\alloc = \alloc(\bs(\bv))$ has three possibilities:
\begin{itemize}
    \item {\bf Case~(I).}
    The allocated bidder $\alloc$ has a normal bid $s_{\alloc}(v_{\alloc}) > \gamma$ and a normal value $v_{\alloc} > \gamma$.
    
    Regarding the coupling between both strategy profiles $\bs^{*} = \{s_{i}^{*}\}_{i \in [n]}$ and $\bs = \{s_{i}\}_{i \in [n]}$, normal bidders $\{i \in [n] \mid s_{i}(v_{i}) > \gamma\}$ preserve their bids $s_{i}^{*}(v_{i}) = s_{i}(v_{i})$, while low/boundary bidders $\{j \in [n] \mid s_{j}(v_{j}) \leq \gamma\}$ never increase their bids $s_{j}^{*}(v_{j}) \leq s_{j}(v_{j})$.
    
    In the coupled scenario $(\alloc^{*},\, \bs^{*}(\bv))$, bidder $\alloc$ still has the first-order bid $s_{\alloc}^{*}(v_{\alloc}) = \max(\bs^{*}(\bv))$ and, (\Cref{lem:bid_distribution:continuity} of \Cref{lem:bid_distribution}) almost surely, is the only first-order bidder $\argmax(\bs^{*}(\bv)) = \{\alloc\}$ and thus keeps winning $\alloc^{*}(\bs^{*}(\bv)) = \alloc$.

    \item {\bf Case~(II).}
    The allocated bidder $\alloc$ has a boundary bid $s_{\alloc}(v_{\alloc}) = \gamma$ and a normal value $v_{\alloc} > \gamma$.
    
    Bidder $\alloc$ is the {\em unique} monopolist (\Cref{lem:monopolist}). Further, other bidders $i \in [n] \setminus \{\alloc\}$ have low/boundary bids $s_{i}(v_{i}) \leq \gamma$ and low/boundary values $v_{i} \leq \gamma$ (cf.\ the $\bs^{*}$-construction table).
    
    In the coupled scenario $(\alloc^{*},\, \bs^{*}(\bv))$, bidder $\alloc$ preserves her bid $s_{\alloc}^{*}(v_{\alloc}) = s_{\alloc}(v_{\alloc}) = \gamma$, while other bidders $i \in [n] \setminus \{\alloc\}$ decrease their bids $s_{i}^{*}(v_{i}) \leq s_{i}(v_{i}) - \delta / 2 \leq \gamma - \delta / 2$. Thus, bidder $\alloc$ is the only first order bidder $\argmax(\bs^{*}(\bv)) = \{\alloc\}$ and thus keeps winning $\alloc^{*}(\bs^{*}(\bv)) = \alloc$.
    
    \item {\bf Case~(III).}
    The allocated bidder $\alloc$ has a boundary bid/value $s_{\alloc}(v_{\alloc}) = v_{\alloc} = \gamma$.
    
    The original scenario has {\em no} monopolist (\Cref{lem:monopolist}), almost surely. All bidders $i \in [n]$ have low/boundary bids $s_{i}(v_{i}) \leq \gamma$ and low/boundary values $v_{i} \leq \gamma$ (cf.\ the $\bs^{*}$-construction table). Further, bidders $N_{\mathrm{BDY}} = \{j \in [n] \mid s_{j}(v_{j}) = v_{j} = \gamma\}$, including the allocated bidder $\alloc$, have boundary bids/values.
    
    In the coupled scenario $(\alloc^{*},\, \bs^{*}(\bv))$, bidders $j \in N_{\mathrm{BDY}}$ have the bid $s_{j}^{*}(v_{j}) = s_{j}(v_{j}) - \delta / 2 = \gamma - \delta / 2$, while the other bidders $i \in [n] \setminus N_{\mathrm{BDY}}$ have lower bids $s_{i}^{*}(v_{i}) = s_{i}(v_{i}) - \delta < \gamma - \delta$. Thus, bidders $j \in N_{\mathrm{BDY}}$ are exactly the coupled first-order bidders $\argmax(\bs^{*}(\bv)) = N_{\mathrm{BDY}}$. I.e., regardless of the tie-breaking rule, any possible allocation $\alloc^{*} = \alloc^{*}(\bs^{*}(\bv)) \in N_{\mathrm{BDY}}$ always realizes the boundary {\SocialWelfare} $v_{\alloc^{*}} = \gamma$, the same as the original scenario $v_{\alloc} = \gamma$.
\end{itemize}
So, the coupling between $\bs^{*} = \{s_{i}^{*}\}_{i \in [n]}$ and $\bs = \{s_{i}\}_{i \in [n]}$ preserves the same auction {\SocialWelfare}, almost surely. In expectation, we have $\FPA(\bV,\, \alloc^{*},\, \bs^{*}) = \FPA(\bV,\, \alloc,\, \bs)$.

\vspace{.1in}
\noindent
{\bf \Cref{thm:apx_bne:universality}.}
Under our coupling: The normal bids $s_{i}(v_{i}) > \gamma$ keep the same $s_{i}^{*}(v_{i}) = s_{i}(v_{i})$. The low bids $s_{i}(v_{i}) < \gamma$ are shifted by a $-\delta$ distance, namely $s_{i}^{*}(v_{i}) = s_{i}(v_{i}) - \delta$. Moreover, the boundary bids $s_{i}(v_{i}) = \gamma$ are ``split'' into (i)~$s_{i}^{*}(v_{i}) = \gamma - \delta$ for low values $v_{i} < \gamma$, (ii)~$s_{i}^{*}(v_{i}) = \gamma - \delta / 2$ for boundary values $v_{i} = \gamma$, or (iii)~$s_{i}^{*}(v_{i}) = \gamma$ for high values $v_{i} > \gamma$, which occurs only for the unique monopolist (if existential).

The coupled infimum first-order bid $\gamma^{*} = \inf(\supp(\calB^{*}))$ is bounded between $\gamma^{*} \in [\gamma - \delta / 2,\, \gamma]$, since in the original scenario, conditioned on the boundary first-order bid $\{\max(\bs(\bv)) = \gamma\}$, there always exists at least one boundary/normal valuer $\{\max(\bv) \geq \gamma\}$ (\Cref{lem:monopolist}).
We would verify the $\delta$-approximate equilibrium conditions, through on case analysis about the original scenario:
\begin{itemize}
    \item A low/boundary value/bid $v_{i} \leq \gamma$ and $s_{i}(v_{i}) \leq \gamma$, with at least one strict inequality.
    
    By construction, the coupled bid $s_{i}^{*}(v_{i}) \leq \gamma - \delta$ is strictly below the coupled infimum first-order bid $\gamma^{*} \in [\gamma - \delta / 2,\, \gamma]$, yielding a zero interim allocation/utility $= 0$. In contrast, because this bidder $i$ has a low/boundary value $v_{i} \leq \gamma$, any deviation bid $b^{*} \geq 0$ yields an interim utility at most $\leq \min(v_{i} - \gamma^{*},\, 0) \leq \delta / 2 < \delta$.
    
    \item A boundary value/bid $v_{i} = s_{i}(v_{i}) = \gamma$.
    
    By construction, the coupled bid $s_{i}^{*}(v_{i}) = \gamma - \delta / 2 < v_{i} = \gamma$ yields a nonnegative interim utility $\geq 0$. In contrast, because this bidder $i$ has a boundary value $v_{i} = \gamma$, any deviation bid $b^{*} \geq 0$ yields an interim utility at most $\leq v_{i} - \gamma^{*} = \delta / 2 < \delta$.
    
    \item A normal value $v_{i} > \gamma$ and a boundary/normal bid $s_{i}(v_{i}) \geq \gamma$.
    
    Following \Cref{lem:bid_distribution:continuity} of \Cref{lem:bid_distribution}, the coupled bid $s_{i}^{*}(v_{i}) = s_{i}(v_{i})$ yields the same nonnegative interim allocation/utility $\geq 0$ as in the original scenario $(\alloc,\, \bs(\bv))$. In comparison: (i)~Any deviation bid $b^{*} \geq \gamma$ yields a smaller or equal interim utility, as a consequence of the exact equilibrium $\bs \in \bbBNE(\bV,\, \alloc)$. (ii)~Any deviation bid $b^{*} < \gamma^{*}$ yields a zero interim utility $= 0$. (iii)~Any deviation bid $b^{*} \in [\gamma^{*},\, \gamma)$ yields at most ``the bid-$\gamma$ interim utility $\leq$ the current interim utility'' plus ``a term of $\gamma - b^{*} \leq \gamma - \gamma^{*} = \delta / 2 < \delta$''.
\end{itemize}
Hence, the coupled strategy profile forms a $\delta$-approximate equilibrium $\bs^{*} \in \bbBNE(\bV,\, \alloc^{*},\, \delta)$.

A minor issue is the above modification may incur negative bids if the original infimum first-order bid is too small $\gamma < \delta$. Instead, we can first slightly shift the original strategies $\bs = \{s_{i}\}_{i \in [n]}$ by a $+\delta$ distance and then reapply the above modification. As a consequence, everything keeps the same, except that the bidders' utilities each drop by a $\delta$ amount. This finishes the proof.
\end{proof}

A revelation of \Cref{thm:apx_bne} is that
we can focus on {\em exact} equilibria $\bs$ in studying the {\PoA}/{\PoS} problems, as if we can control the tie-breaking rule $\alloc$ and choose a {\em compatible} one $\bbBNE(\bV,\, \alloc) \neq \emptyset$. When an {\em incompatible} tie-breaking rule $\alloc^{*}$ are really considered, up to any precision $\delta > 0$, we can still obtain a $\delta$-approximate equilibrium $\bs^{*} \in \bbBNE(\bV,\, \alloc^{*},\, \delta)$ by modifying any ``compatible'' exact equilibrium $\bs \in \bbBNE(\bV,\, \alloc)$. We will adopt this convention in \Cref{sec:BNE_independent,sec:BNE_correlated}, since it simplifies the notation and (essentially) incurs no loss of generality.

\section{Bayesian Nash Equilibria for Independent Valuations}
\label{sec:BNE_independent}

In this section, we prove the following tight {\PoS} result for the canonical setting, namely {\BayesNashEquilibria} for independent valuations.

\begin{theorem}[Tight {\PoS}]
\label{thm:BNE_independent}
Regarding {\BayesNashEquilibria} for correlated valuations, 
the {\PriceofStability} is $1 - 1 / e^{2} \approx 0.8647$.
\end{theorem}

The lower-bound part of \Cref{thm:BNE_independent} immediately follows from the tight {\PoA} result from \cite{JL22}. To get the upper-bound part, we only need to construct an instance whose {\PoS} is exactly $1 - 1 / e^{2}$. Technically, we provide a sequence of instances whose {\PoS} asymptotically approaches
$1 - 1 / e^{2}$.  The following instances are a slight modification from the tight {\PoA} instances due to \cite[Example~4]{JL22} such that each modified instance has one unique equilibrium.

\begin{example}
\label{exp:BNE_independent}
Given an arbitrarily small constant $\epsilon \in (0,\, 1 / 8)$, consider the $(n + 1)$-bidder instance $\{H\} \cup \{L_{i}\}_{i \in [n]}$ for $n = \lceil 1 / \epsilon \rceil \geq 8$ in terms of value distributions $\bV = V_{H} \otimes \{V_{L}\}^{\otimes n}$.
\begin{itemize}
    \item Bidder $H$ has a Bernoulli random value $v_{H} \sim V_{H}$ that $\Pr[v_{H} = 0] = \epsilon$ and $\Pr[v_{H} = 1] = 1 - \epsilon$.
    
    \item Bidders $\{L_{i}\}_{i \in [n]}$ have i.i.d.\ values $(v_{L,\, i})_{i \in [n]} \sim \{V_{L}\}^{\otimes n}$ whose common value distribution $V_{L}$ is given by the parametric equation $V_{L}(1 - \frac{n - (t - 1)}{n t - (t - 1)} \cdot t^{2} \cdot e^{2 - 2t}) = \sqrt[n]{4 / t^{2} \cdot e^{2t - 4}}$ for $t \in [1,\, 2]$. This value distribution $V_{L}$ is supported on $\supp(V_{L}) = [0,\, 1 - \frac{2n - 2}{2n - 1} \cdot 2 / e^{2}]$ and has a probability mass $V_{L}(0) = \sqrt[n]{4 / e^{2}}$ at the zero value.
\end{itemize}
The considered {\FirstPriceAuction} $\alloc \in \bbFPA$, under the all-zero bid profile $\bb = b_{H} \otimes (b_{L,\, i})_{i \in [n]} = \zeros$, favors bidder $H$, but otherwise is arbitrary.
\end{example}

\subsection{The focal equilibrium}
\label{subsec:BNE_independent:focal}

In this part, we verify that the worst-case equilibrium for the tight {\PoA} instance \cite[Example~4]{JL22}, after a slight adjustment, is still an equilibrium for our modified instance. (The proof is similar to the equilibrium condition analysis in \cite[Section~6]{JL22}.)

For clarity, this modified equilibrium will be called the {\em \textbf{focal}} strategy profile or, after verifying the equilibrium condition, the {\em \textbf{focal}} equilibrium. Let $\lambda^{*} \eqdef 1 - 4 / e^{2} \approx 0.4587$. The focal strategy profile $\bs^{*} = \{s_{H}^{*}\} \otimes \{s_{L}^{*}\}^{\otimes n}$ is given as follows; see \Cref{fig:LB} for a visual aid.
\begin{flushleft}
\begin{itemize}
    \item Bidder $H$ has a (mixed) strategy $s_{H}^{*}(0) \equiv 0$ for a zero value $\{v_{H} = 0\}$ and $s_{H}^{*}(1) \sim S_{H}^{*}$ for a nonzero value $\{v_{H} = 1\}$. Here the bid distribution $S_{H}^{*}$ is given by the implicit equation $b = 1 - 4 \cdot (\epsilon + (1 - \epsilon) \cdot S_{H}^{*}) \cdot e^{2 - 4\sqrt{\epsilon + (1 - \epsilon) \cdot S_{H}^{*}}}$ for $S_{H}^{*} \in [1 - \frac{3 / 4}{1 - \epsilon},\, 1]$.
    
    The random bid $s_{H}^{*}(v_{H})$ is supported on $\{0\} \cup \supp(S_{H}^{*}) = [0,\, 1 - 4 / e^{2}] = [0,\, \lambda^{*}]$.
    
    \item Bidders $\{L_{i}\}_{i \in [n]}$ have (deterministic) identical strategies $\{s_{L}^{*}\}^{\otimes n}$ that are given by the parametric equation $s_{L}^{*}(1 - \tfrac{n - (t - 1)}{n t - (t - 1)} \cdot t^{2} \cdot e^{2 - 2t}) = 1 - t^{2} \cdot e^{2 - 2t}$ for $t \in [1,\, 2]$.
    
    The random bids $s_{L}^{*}(v_{L,\, i})$ are supported on $\{1 - t^{2} \cdot e^{2 - 2t} \mid t \in [1,\, 2]\} = [0,\, 1 - 4 / e^{2}] = [0,\, \lambda^{*}]$.
\end{itemize}
\end{flushleft}

\Cref{lem:BNE_independent:focal} checks the equilibrium condition for the focal strategy profile $\bs^{*}$.

\begin{lemma}[Equilibrium]
\label{lem:BNE_independent:focal}
The following hold for the focal strategy profile $\bs^{*} = \{s_{H}^{*}\} \otimes \{s_{L}^{*}\}^{\otimes n}$:
\begin{enumerate}[font = {\em\bfseries}]
    \item\label{lem:BNE_independent:focal:1}
    Bidder $H$ has the bid distribution $B_{H}^{*}$ given by the implicit equation $b = 1 - 4B_{H}^{*} \cdot e^{2 - 4\sqrt{B_{H}^{*}}}$ for $B_{H}^{*} \in [1 / 4,\, 1]$, and a constant bid-to-value mapping $\varphi_{H}(b) = 1$ for $b \in [0,\, \lambda^{*}]$.
    
    
    \item\label{lem:BNE_independent:focal:2}
    Bidders $\{L_{i}\}_{i \in [n]}$ have identical bid distributions $\{B_{L}^{*}\}^{\otimes n}$ given by $B_{L}^{*}(b) = \sqrt[n]{(1 - \lambda^{*}) / (1 - b)}$ for $b \in [0,\, \lambda^{*}]$, and identical bid-to-value mappings $\{\varphi_{L}^{*}\}^{\otimes n}$ given by the parametric equation $\varphi_{L}^{*}(1 - t^{2} \cdot e^{2 - 2t}) = 1 - \frac{n - (t - 1)}{n t - (t - 1)} \cdot t^{2} \cdot e^{2 - 2t}$ for $t \in [1,\, 2]$.
    
    \item\label{lem:BNE_independent:focal:3}
    The focal strategy profile $\bs^{*} = s_{H}^{*} \otimes \{s_{L}^{*}\}^{\otimes n}$ forms a {\BayesNashEquilibrium} $\bs^{*} \in \bbBNE(\bV)$.
\end{enumerate}
\end{lemma}

\begin{figure}[t]
    \centering
    \includegraphics[width = .9\textwidth]{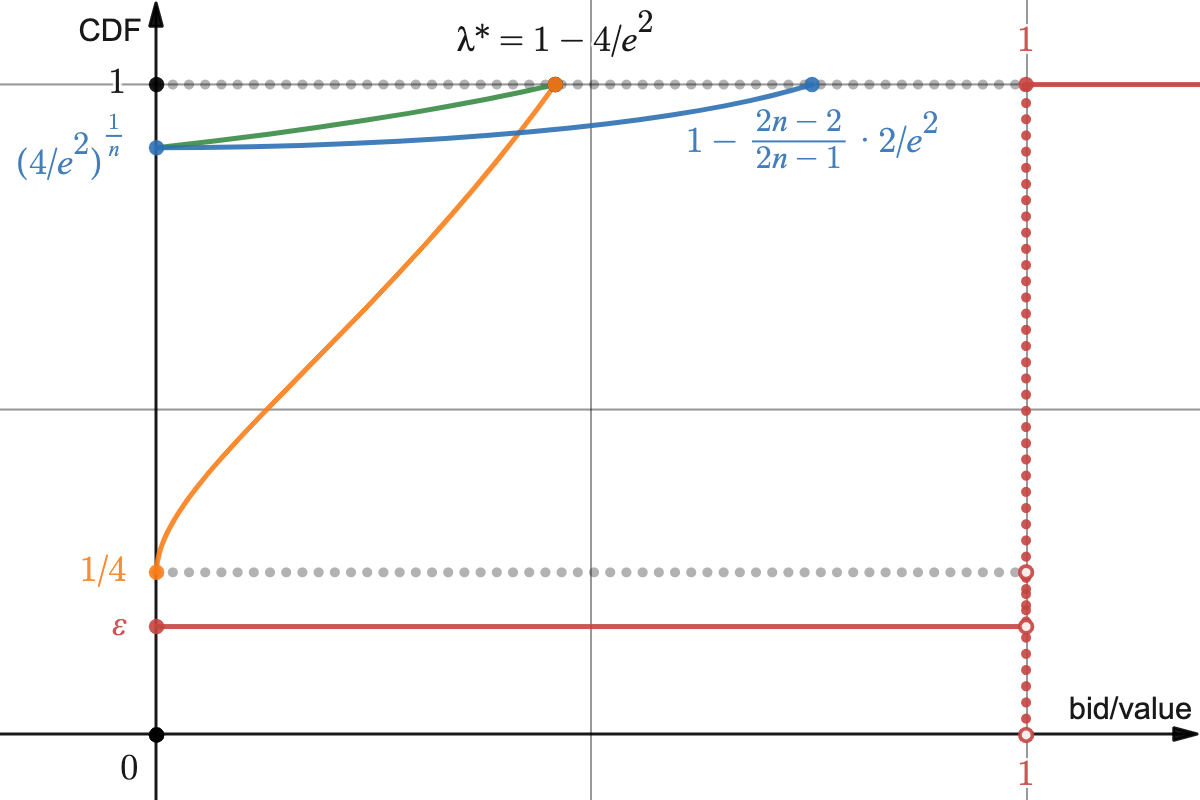}
    \caption{Demonstration for the $(n + 1)$-bidder instance $\bV = V_{H} \otimes \{V_{L}\}^{\otimes n}$ in \Cref{exp:BNE_independent}. \\
    (red) Bidder $H$ has a Bernoulli value distribution $V_{H}(v) = \epsilon$ for $v \in [0,\, 1)$ and $V_{H}(v) = 1$ for $v \geq 1$. \\
    (blue) Bidders $\{L_{i}\}_{i \in [n]}$ have identical value distributions $\{V_{L}\}^{\otimes n}$ given by the parametric equation $V_{L}(1 - \frac{n - (t - 1)}{n t - (t - 1)} \cdot t^{2} \cdot e^{2 - 2t}) = \sqrt[n]{4 / t^{2} \cdot e^{2t - 4}}$ for $t \in [1,\, 2]$. \\
    Under the focal strategy profile $\bs^{*} = \{s_{H}^{*}\} \otimes \{s_{L}^{*}\}^{\otimes n}$: Over the support $b \in [0,\, \lambda^{*}]$, the resulting bid distributions $\bB^{*} = B_{H}^{*} \otimes \{B_{L}^{*}\}^{\otimes n}$ are given by
    (orange) the implicit equation $b = 1 - 4B_{H}^{*} \cdot e^{2 - 4\sqrt{B_{H}^{*}}}$ for $B_{H}^{*} \in [1 / 4,\, 1]$ and
    (green) the parametric equation $B_{L}^{*}(1 - t^{2} \cdot e^{2 - 2t}) = \sqrt[n]{4 / t^{2} \cdot e^{2t - 4}}$ for $t \in [1,\, 2]$.}
    \label{fig:LB}
\end{figure}

\begin{proof}
We first reason about {\bf \Cref{lem:BNE_independent:focal:1}} and {\bf \Cref{lem:BNE_independent:focal:2}}.

For bidder $H$, the strategy $s_{H}^{*}$ converts (i)~all densities at the zero value $\Pr[v_{H} = 0] = \epsilon$ to densities at the zero bid $s_{H}^{*}(0) \equiv 0$ and (ii)~all densities at the nonzero value $\Pr[v_{H} = 1] = 1 - \epsilon$ to densities that follow the bid distribution $S_{H}^{*}$, which is supported on $\supp(S_{H}^{*}) = [0,\, \lambda^{*}]$.
Overall, the bid distribution $s_{H}^{*}(v_{H}) \sim B_{H}^{*}$ can be written as $B_{H}^{*}(b) = \epsilon + (1 - \epsilon) \cdot S_{H}^{*}(b)$, over the bid support $b \in [0,\, \lambda^{*}]$. Plugging this formula into the defining implicit equation for $S_{H}^{*}$, we can conclude with $b = 1 - 4B_{H}^{*} \cdot e^{2 - 4\sqrt{B_{H}^{*}}}$ for $B_{H}^{*} \in [1 / 4,\, 1]$, as desired.

For bidders $\{L_{i}\}_{i \in [n]}$ and their identical strategies $\{s_{L}^{*}\}^{\otimes n}$, the value formula $1 - \tfrac{n - (t - 1)}{n t - (t - 1)} \cdot t^{2} \cdot e^{2 - 2t}$ and the bid formula $1 - t^{2} \cdot e^{2 - 2t}$ both are increasing in $t \in [1,\, 2]$. Thus, the identical bid distributions $\{B_{L}^{*}\}^{\otimes n}$ can be written as $\{(b,\, B_{L}^{*}) = (s_{L}^{*}(v),\, V_{L}(v)) \mid v \in \supp(V_{L}) = [0,\, 1 - \frac{2n - 2}{2n - 1} \cdot 2 / e^{2}]\}$. Plugging the defining parametric equations for $s_{L}^{*}$ and $V_{L}$ into this formula, those bid distributions $\{B_{L}^{*}\}^{\otimes n}$ can be formulated as $\{(b,\, B_{L}^{*}) = (1 - t^{2} \cdot e^{2 - 2t},~ \sqrt[n]{4 / t^{2} \cdot e^{2t - 4}}) \mid t \in [1,\, 2]\}$. After rearranging, we can conclude with $B_{L}^{*}(b) = \sqrt[n]{(1 - \lambda^{*}) / (1 - b)}$ for $b \in [0,\, \lambda^{*}]$, as desired.

Bidder $H$ competes with bidders $\{L_{i}\}_{i \in [n]}$, thus having the competing bid distribution $\calB_{-H}^{*}(b) = \big(B_{L}^{*}(b)\big)^{n} = (1 - \lambda^{*}) / (1 - b)$ and the {\em constant} bid-to-value mapping $\varphi_{H}^{*}(b) = b + \calB_{-H}^{*}(b) / {\calB_{-H}^{*}}'(b) = 1$ for $b \in [0,\, \lambda^{*}]$, as desired.




Let $t \eqdef 2\sqrt{B_{H}^{*}} \in [1,\, 2]$. Then we have $b = 1 - 4B_{H}^{*} \cdot e^{2 - 4\sqrt{B_{H}^{*}}} = 1 - t^2 \cdot e^{2 - 2t}$ and the derivative $\frac{\d b}{\d t} = (2t^{2} - 2t) \cdot e^{2 - 2t}$. Each bidder $L_{i}$ for $i \in [n]$ competes with bidders $\{H\} \cup \{L_{j}\}_{j \in [n] \setminus \{i\}}$, hence the competing bid distribution $\calB_{-L}^{*}(b) = B_{H}^{*}(b) \cdot \big(B_{L}^{*}(b)\big)^{n - 1}$. In terms of the parameter $t \in [1,\, 2]$, we can substitute $B_{H}^{*} = t^{2} / 4$, $\lambda^{*} = 1 - 4 / e^{2}$, and $b = 1 - t^2 \cdot e^{2 - 2t}$, rewriting
\[
    \calB_{-L}^{*} ~=~ B_{H}^{*} \cdot \big(B_{L}^{*}\big)^{n - 1}
    ~=~ B_{H}^{*} \cdot \Big(\frac{1 - \lambda^{*}}{1 - b}\Big)^{\frac{n - 1}{n}}
    ~=~ t^{2} / 4 \cdot \Big(\frac{4 / t^2}{e^{4 - 2t}}\Big)^{\frac{n - 1}{n}}.
\]
Then in terms of $t \in [1,\, 2]$, the bid-to-value mapping $\varphi_{L}^{*}(b) = b + \frac{\calB_{-L}^{*}}{\d \calB_{-L}^{*} / \d b}$ is given by
\begin{align*}
    \varphi_{L}^{*}
    & ~=~ b + \calB_{-L}^{*} \cdot \frac{\d b / \d t}{\d \calB_{-L}^{*} / \d t} \\
    & ~=~ (1 - t^2 \cdot e^{2 - 2t}) + n t \cdot \frac{(t^{2} - t) \cdot e^{2 - 2t}}{n t - (t - 1)} \phantom{\Big.} \\
    & ~=~ 1 - \frac{n - (t - 1)}{n t - (t - 1)} \cdot t^{2} \cdot e^{2 - 2t}. \phantom{\Big.}
\end{align*}
Hence, we obtain the parametric equation $\varphi_{L}^{*}(1 - t^{2} \cdot e^{2 - 2t}) = 1 - \frac{n - (t - 1)}{n t - (t - 1)} \cdot t^{2} \cdot e^{2 - 2t}$ for $t \in [1,\, 2]$, as desired. Notably, this bid-to-value mapping is the inverse function of $\{L_{i}\}_{i \in [n]}$'s focal strategies $s_{L}^{*}(1 - \tfrac{n - (t - 1)}{n t - (t - 1)} \cdot t^{2} \cdot e^{2 - 2t}) = 1 - t^{2} \cdot e^{2 - 2t}$ for $t \in [1,\, 2]$.\footnote{It is easy to check that both formulas $1 - \tfrac{n - (t - 1)}{n t - (t - 1)} \cdot t^{2} \cdot e^{2 - 2t}$ and $1 - t^{2} \cdot e^{2 - 2t}$ are strictly increasing in $t \in [1,\, 2]$.}
{\bf \Cref{lem:BNE_independent:focal:1}} and {\bf \Cref{lem:BNE_independent:focal:2}} follow then.

\vspace{.1in}
\noindent
{\bf \Cref{lem:BNE_independent:focal:3}.}
Bidder $H$ has the interim utility formula $u_{H}^{*}(v_{H},\, b) = (v_{H} - b) \cdot \calB_{-H}^{*}(b) = (v_{H} - b) \cdot \frac{1 - \lambda^{*}}{1 - b}$ for $b \in [0,\, \lambda^{*}]$. Under a zero value $\{v_{H} = 0\}$, clearly the focal bid $s_{H}^{*}(0) \equiv 0$ must be utility-optimal. Under a nonzero value $\{v_{H} = 1\}$, all bids $b \in [0,\, \lambda^{*}]$ yield the same interim utility $u_{H}^{*}(1,\, b) = 1 - \lambda^{*}$, so the focal bid $s_{H}^{*}(1) \sim S_{H}^{*}$ that $\supp(S_{H}^{*}) = [0,\, \lambda^{*}]$ also is utility-optimal.

Moreover, bidders $\{L_{i}\}_{i \in [n]}$ have the same interim utility formula $u_{L}^{*}(v_{L},\, b) = (v_{L} - b) \cdot \calB_{-L}^{*}(b)$ for $b \in [0,\, \lambda^{*}]$. For any given value $v_{L} \in [0,\, 1 - \frac{2n - 2}{2n - 1} \cdot 2 / e^{2}]$, a bid $b \in [0,\, \lambda^{*}]$ is utility-optimal when it satisfies that $0 = \frac{\partial}{\partial b} u_{L}^{*}(v_{L},\, b) = -\calB_{-L}^{*}(b) + (v_{L} - b) \cdot {\calB_{-L}^{*}}'(b)$, or equivalently, that $v_{L} = \varphi_{L}^{*}(b)$.
The focal bid $s_{L}^{*}(v_{L})$ {\em is} utility optimal, namely $v_{L} = \varphi_{L}^{*}(s_{L}^{*}(v_{L}))$, because the bid-to-value mapping $\varphi_{L}^{*}$ is the inverse function of the focal strategy $s_{L}^{*}$.

Thus, all bidders $\{H\} \cup \{L_{i}\}_{i \in [n]}$ meet the equilibrium conditions. {\bf \Cref{lem:BNE_independent:focal:3}} follows then.
\end{proof}

Below, \Cref{lem:BNE_independent:efficiency} measures the expected optimal/auction {\SocialWelfares} from our $(n + 1)$-bidder instance $\bV = V_{H} \otimes \{V_{L}\}^{\otimes n}$ at the focal equilibrium $\bs^{*} = \{s_{H}^{*}\} \otimes \{s_{L}^{*}\}^{\otimes n}$.
The proof relies on the auction {\SocialWelfare} formula from \cite[Lemma~2.20]{JL22}.

\begin{lemma}[Auction {\SocialWelfare} {\cite{JL22}}]
\label{lem:auction_welfare}
The expected auction {\SocialWelfare} $\FPA(\bV,\, \bs)$ at a {\BayesNashEquilibrium} $\bs \in \bbBNE(\bV)$, on having a monopolist $H$, can be formulated as follows:
\begin{align*}
    \FPA(\bV,\, \bs)
    ~=~ \Ex_{v_{H},\, s_{H}}[ v_{H} \mid s_{H}(v_{H}) = \gamma ] \cdot \calB(\gamma) + \sum_{i \in [n]} \bigg(\int_{\gamma}^{\lambda} \varphi_{i}(b) \cdot \frac{B'_{i}(b)}{B_{i}(b)} \cdot \calB(b) \cdot \d b\bigg).
\end{align*}
\end{lemma}

\begin{lemma}[Efficiency]
\label{lem:BNE_independent:efficiency}
The following hold for the focal equilibrium $\bs^{*} = \{s_{H}^{*}\} \otimes \{s_{L}^{*}\}^{\otimes n}$:
\begin{enumerate}[font = {\em\bfseries}]
    \item\label{lem:BNE_independent:efficiency:1}
    The expected optimal {\SocialWelfare} $\OPT(\bV,\, \bs^{*}) \geq 1 - \epsilon$.

    \item\label{lem:BNE_independent:efficiency:2}
    The expected auction {\SocialWelfare} $\FPA(\bV,\, \bs^{*}) \leq 1 - (1 - \epsilon) \cdot e^{-2}$.
\end{enumerate}
\end{lemma}

\begin{proof}
Let us prove {\bf \Cref{lem:BNE_independent:efficiency:1,lem:BNE_independent:efficiency:2}} one by one.

\vspace{.1in}
\noindent
{\bf \Cref{lem:BNE_independent:efficiency:1}.}
The realized optimal {\SocialWelfare} is at least bidder $H$'s realized value $v_{H} \sim V_{H}$, namely a Bernoulli random value $\Pr [v_{H} = 0] = \epsilon$ and $\Pr [v_{H} = 1] = 1 - \epsilon$ (\Cref{exp:BNE_independent}). In expectation, we have $\OPT(\bV,\, \bs^{*}) \geq \E[v_{H}] = 1 - \epsilon$. {\bf \Cref{lem:BNE_independent:efficiency:1}} follows then.

\vspace{.1in}
\noindent
{\bf \Cref{lem:BNE_independent:efficiency:2}.}
Following \Cref{lem:auction_welfare}, with the focal first-order bid CDF $\calB^{*}(b) = B_{H}^{*}(b) \cdot \big(B_{L}^{*}(b)\big)^{n}$, the expected auction {\SocialWelfare} $\FPA(\bV,\, \bs^{*})$ from our $(n + 1)$-bidder instance is given by
\begin{align*}
    & \phantom{~=~} \FPA(\bV,\, \bs^{*}) \\
    & ~=~ \Ex_{v_{H},\, s_{H}^{*}}[ v_{H} \mid s_{H}^{*}(v_{H}) = 0 ] \cdot \calB^{*}(0) + \int_{0}^{\lambda^{*}} \Big(\varphi_{H}^{*}(b) \cdot \frac{{B_{H}^{*}}'(b)}{B_{H}^{*}(b)} \cdot \calB^{*}(0) + n \cdot \varphi_{L}^{*}(b) \cdot \frac{{B_{L}^{*}}'(b)}{B_{L}^{*}(b)} \cdot \calB^{*}(b)\Big) \cdot \d b \\
    & ~\leq~ \calB^{*}(0) + \int_{0}^{\lambda^{*}} \Big(\frac{{B_{H}^{*}}'(b)}{B_{H}^{*}(b)} \cdot \calB^{*}(0) + n \cdot \varphi_{L}^{*}(b) \cdot \frac{{B_{L}^{*}}'(b)}{B_{L}^{*}(b)} \cdot \calB^{*}(b)\Big) \cdot \d b \\
    & ~=~ \calB^{*}(0) + \int_{0}^{\lambda^{*}} {\calB^{*}}'(b) \cdot \d b - \int_{0}^{\lambda^{*}} n \cdot \big(1 - \varphi_{L}^{*}(b)\big) \cdot \frac{{B_{L}^{*}}'(b)}{B_{L}^{*}(b)} \cdot \calB^{*}(b) \cdot \d b \\
    & ~=~ 1 - \int_{0}^{\lambda^{*}} n \cdot \big(1 - \varphi_{L}(b)\big) \cdot \frac{1}{n \cdot (1 - b)} \cdot B_{H}^{*}(b) \cdot \frac{1 - \lambda^{*}}{1 - b} \cdot \d b \\
    & ~=~ 1 - (1 - \lambda^{*}) \cdot \int_{0}^{\lambda^{*}} \big(1 - \varphi_{L}^{*}(b)\big) \cdot \frac{B_{H}^{*}(b)}{(1 - b)^2} \cdot \d b \\
    & ~=~ 1 - (1 - \lambda^{*}) \cdot \int_{1}^{2} \Big(\frac{n - (t - 1)}{n t - (t - 1)} \cdot t^{2} \cdot e^{2 - 2t}\Big) \cdot \frac{t^{2} / 4}{(t^{2} \cdot e^{2 - 2t})^{2}} \cdot \Big(\frac{\d b}{\d t}\Big) \cdot \d t \\
    & ~=~ 1 - (1 - \lambda^{*}) \cdot \int_{1}^{2} \Big(\frac{n - (t - 1)}{n t - (t - 1)} \cdot t^{2} \cdot e^{2 - 2t}\Big) \cdot \frac{t^{2} / 4}{(t^{2} \cdot e^{2 - 2t})^{2}} \cdot (2t^{2} - 2t) \cdot e^{2 - 2t} \cdot \d t \\
    & ~=~ 1 - (1 - \lambda^{*}) \cdot \int_{1}^{2} \frac{n - (t - 1)}{n t - (t - 1)} \cdot \frac{t^{2} - t}{2} \cdot \d t \\
    & ~\leq~ 1 - (1 - \lambda^{*}) \cdot \int_{1}^{2} \Big(1 - \frac{1}{n}\Big) \cdot \frac{t - 1}{2} \cdot \d t \\
    & ~=~ 1 - \Big(1 - \frac{1}{n}\Big) \cdot e^{-2} \phantom{\Big.} \\
    & ~\leq~ 1 - (1 - \epsilon) \cdot e^{-2}. \phantom{\Big.}
\end{align*}
The definition of $t \in [1,\, 2]$ and the expressions of $b$, $B_{H}^{*}$, $\varphi_{L}^{*}$ and $\frac{\d b}{\d t}$ in terms of $t$ are given in the proof of \Cref{lem:BNE_independent:efficiency}.
{\bf \Cref{lem:BNE_independent:efficiency:2}} follows then.
This finishes the proof of \Cref{lem:BNE_independent:efficiency}.
\end{proof}

\subsection{Uniqueness of equilibria}
\label{subsec:BNE_independent:unqiue}

In this part, we show that the focal equilibrium $\bs^{*} = \{s_{H}^{*}\} \otimes \{s_{L}^{*}\}^{\otimes n}$ is the {\em unique} equilibrium for our modified instance $\bV = V_{H} \otimes \{V_{L}\}^{\otimes n}$; thus the tight {\PoS} bound is the same as the tight {\PoA} bound. This uniqueness is the key ingredient of our {\PoS} characterization. To have a better sense, let us explain the high-level ideas before giving the formal proof.

We first prove that bidders $\{L_{i}\}_{i \in [n]}$ must have identical strategies because they have identical value distributions (\Cref{lem:BNE_independent:symmetry}). Thus, we ``truly'' have just two kinds of bidders, $H$ versus $\{L\}^{\otimes n}$. Then an equilibrium can be obtained by resolving an ordinal differential equation (ODE) in terms of the bid distributions for $H$ and $L$ (\cref{lem:BNE_independent:unique}). Once the boundary conditions are specified, an ODE ``usually'' has one unique solution. Essentially, the possible non-uniqueness of equilibria stems from different boundary conditions.

The main technical part is to uniquely determine the boundary condition, namely every bidder $H$ or $L$ must have her bid support being exactly the interval $[0,\, \lambda^{*}]$. Compared with the tight {\PoA} instances \cite[Example~4]{JL22}, we modify bidder $H$'s value distribution by putting a tiny probability mass at the zero value $\Pr[v_{H} = 0] = \epsilon$. In this way, every bidder $H$ or $L$ is enforced a zero bid, once this bidder has a zero value (\Cref{lem:BNE_independent:infimum}). Then, it is easy to conclude that the identical bid support of bidders $\{L\}^{\otimes n}$ is exactly an interval $[0,\, \lambda]$ -- having densities {\em almost everywhere} -- since those bidders have an {\em uninterrupted} value support $\supp(V_{L}) = [0,\, 1 - \frac{2n - 2}{2n - 1} \cdot 2 / e^{2}]$ (\Cref{lem:BNE_independent:L}); but whether this supremum bid $\lambda$ is exactly the $\lambda^{*} = 1 - 4 / e^{2} \approx 0.4587$ is still unclear.

However, determining the desirable boundary condition for bidder $H$ is highly nontrivial. This bidder has an {\em interrupted} value support $\supp(V_{H}) = \{0,\, 1\}$, so the above arguments fail to work. Instead, we first show that bidder $H$'s bid support is the union $\{0\} \cup [\mu,\, \lambda]$ of (i)~the zero bid $\{0\}$, which corresponds to the zero value $\Pr[v_{H} = 0] = \epsilon$; and (ii)~an interval $[\mu,\, \lambda]$ for some $\mu \geq 0$ -- having densities {\em almost everywhere} -- which corresponds to the nonzero value $\Pr[v_{H} = 1] = 1 - \epsilon$.\footnote{A better interpretation is from the perspective of quantiles; then there is no ambiguity even in the case $\mu = 0$.}
The undesirable case $\mu > 0$ is really possible, if we could slightly adjust \Cref{exp:BNE_independent}, e.g., changing the ``success $= \epsilon$''/``failure $= 1 - \epsilon$'' probabilities of bidder $H$'s Bernoulli random value. But under our particular construction, only the desirable case $\mu = 0$ turns out to be possible; thus bidder $H$ also has densities almost everywhere on the interval $[0,\, \lambda]$ (\cref{lem:BNE_independent:H}).

Provided with the desirable boundary conditions, we resolve the mentioned ODE, thus {\em uniquely} determining the bid support $[0,\, \lambda] = [0,\, \lambda^{*}]$ and the equilibrium -- precisely the focal equilibrium $\bs^{*} = \{s_{H}^{*}\} \otimes \{s_{L}^{*}\}^{\otimes n}$ (\Cref{lem:BNE_independent:unique}).

In the rest of \Cref{subsec:BNE_independent:unqiue}, we start with a generic equilibrium $\bs = \{s_{H}\} \otimes \{s_{L,\, i}\}_{i \in [n]} \in \bbBNE(\bV)$ and present the formal proof.

\begin{lemma}
\label{lem:BNE_independent:infimum}
Each bidder $\sigma \in \{H\} \cup \{L_{i}\}_{i \in [n]}$, on having a zero value $\{v_{\sigma} = 0\}$, takes a zero bid $s_{\sigma}(v_{\sigma}) = 0$ almost surely. Hence, bidder $H$ takes a zero bid with probability $B_{H}(0) \geq V_{H}(0) = \epsilon$ and each bidder $L_{i}$ for $i \in [n]$ takes a zero bid with probability $B_{L,\, i}(0) \geq V_{L}(0) = \sqrt[n]{4 / e^{2}}$. Further, the infimum bid $\gamma = \inf(\supp(\calB))$ is zero $\gamma = 0$.
\end{lemma}

\begin{proof}
The all-zero value profile $\{\bv = \zeros\}$ occurs with probability $V_{H}(0) \cdot (V_{L}(0))^{n} = \epsilon \cdot 4 / e^{2} > 0$.
Conditioned on this, the bid profile must also be all-zero $\{\bs(\bv) = \zeros\}$, almost surely --
Otherwise, with a nonzero probability, the allocated bidder $\alloc = \alloc(\bs(\bv))$ gains a {\em strictly} negative utility $< 0$ since she has a nonzero bid $s_{\alloc}(v_{\alloc}) > 0$ and a zero value $v_{\alloc} = 0$, which contradicts the equilibrium condition (\Cref{def:bne_formal}).
We are considering {\BayesNashEquilibria}, namely the strategies $s_{\sigma}(v_{\sigma})$ for $\sigma \in \{H\} \cup \{L_{i}\}_{i \in [n]}$ only depend on individual values $v_{\sigma}$. Therefore, for each individual bidder, a zero value $\{v_{\sigma} = 0\}$ enforces a zero bid $\{s_{\sigma}(v_{\sigma}) = 0\}$, almost surely.
\end{proof}

\begin{lemma}
\label{lem:BNE_independent:symmetry}
Bidders $\{L_{i}\}_{i \in [n]}$ play identical strategies $\{s_{L,\, i}\}_{i} = \{s_{L}\}^{\otimes n}$ everywhere except on a zero-measure set of values. Hence, bid distributions $\{B_{L,\, i}\}_{i \in [n]} = \{B_{L}\}^{\otimes n}$ are identical and the bid-to-value mappings $\{\varphi_{L,\, i}\}_{i \in [n]} = \{\varphi_{L}\}^{\otimes n}$ are identical.
\end{lemma}

\begin{proof}
\Cref{lem:BNE_independent:symmetry} is almost a direct implication of \cite[Corollary~3.2]{CH13}, which claims the same result for any subset of bidders $J = \{j_{1},\, \dots,\, j_{k}\} \subseteq [m]$ in an $m$-bidder {\FirstPriceAuction} that have {\em identical} value distributions $V_{j_{1}} \equiv \dots \equiv V_{j_{k}}$.
The only issue is that \cite[Corollary~3.2]{CH13} requires a tie-breaking rule $\alloc_{J} \in \bbFPA$ that is {\em symmetric} for those bidders $J$. However, regarding our instance $\{H\} \cup \{L_{i}\}_{i \in [n]}$ and tie-breaking rule $\alloc \in \bbFPA$ in \Cref{exp:BNE_independent}: \\
(i)~Tie-breaks at nonzero first-order bids $\{\max(\bs(\bv)) > 0\}$ never occur, almost surely (\Cref{lem:bid_distribution}). \\
(ii)~Tie-breaks at a zero first-order bid $\{\max(\bs(\bv)) = 0\}$, i.e., at the all-zero bid profile $\{\bs(\bv) = 0\}$, always favor bidder $H$ and thus {\em is} symmetric for bidders $\{L_{i}\}_{i \in [n]}$. \\
Accordingly, the symmetry requirement on the tie-breaking rule fails just for a {\em zero-measure} set of values. Clearly, we can readopt the arguments for \cite[Corollary~3.2]{CH13} to derive \Cref{lem:BNE_independent:symmetry}.
\end{proof}

Recall that under the focal equilibrium $\bs^{*}$, bidders $\{L_{i}\}_{i \in [n]}$ are non-monopoly bidders and have the same probability masses  $B_{L}^{*}(0) = B_{L}(0) = \sqrt[n]{4 / e^{2}}$. Below we show that this also holds for the considered equilibrium $\bs = \{s_{H}\} \otimes \{s_{L}\}^{\otimes n}$.

\begin{lemma}
\label{lem:BNE_independent:L}
Bidders $\{L_{i}\}_{i \in [n]}$ are non-monopoly bidders. Hence, the common bid distribution $B_{L}$ has a probability mass $B_{L}(0) = V_{L}(0) = \sqrt[n]{4 / e^{2}}$ at the zero bid and
has densities almost everywhere over the bid support $b \in [0,\, \lambda]$.
\end{lemma}

\begin{proof}
Bidders $\{L_{i}\}_{i \in [n]}$ have identical value/bid distributions $\{V_{L}\}^{\otimes n}$ and $\{B_{L}\}^{\otimes n}$.
According to \Cref{def:monopolist}, they either ALL are non-monopoly bidders $B_{L}(0) = V_{L}(0)$ or ALL are monopolists $B_{L}(0) > V_{L}(0)$. However, \Cref{lem:monopolist} (that there exists at most one monopolist) eliminates the second case. Hence, bidders $\{L_{i}\}_{i \in [n]}$ are non-monopoly bidders $B_{L}(0) = V_{L}(0)$.

For the sake of contradiction, assume that bid distribution $B_{L}$ has no density around some bid $b \in [0,\, \lambda]$. Then, the competing bid distribution $(B_{L}(b))^{n} = \prod_{i \in [n]} B_{L,\, i}(b)$ for bidder $H$ also has no density around this bid $b \in [0,\, \lambda]$. However, this contradicts \Cref{lem:bid_distribution:monotonicity} of \Cref{lem:bid_distribution}. Refuting our assumption finishes the proof of \Cref{lem:BNE_independent:L}.
\end{proof}

\begin{lemma}
\label{lem:BNE_independent:H}
Bidder $H$ is the (unique) monopolist. Hence, bid distribution $B_{H}$ has a probability mass $B_{H}(0) > V_{H}(0) = \epsilon$ at the zero bid and the bid-to-value mapping is constant $\varphi_{H}(b) = 1$ over the bid support $b \in [0,\, \lambda]$.
\end{lemma}

\begin{proof}
Bidder $H$ has (\Cref{exp:BNE_independent}) a Bernoulli random value $\Pr[v_{H} = 0] = \epsilon$ and $\Pr[v_{H} = 1] = 1 - \epsilon$ and (\Cref{lem:high_bid}) an increasing bid-to-value mapping $\varphi_{H}(b)$.

Assume to the contrary that bidder $H$ is a non-monopoly bidder $B_{H}(0) = V_{H}(0) = \epsilon$, namely a zero value $\{v_{H} = 0\}$ induces a zero bid $s_{H}(v_{H}) = 0$ almost surely, and a nonzero value $\{v_{H} = 1\}$ induces a nonzero bid $s_{H}(v_{H}) \in (0,\, \lambda]$ almost surely.

Consider the threshold bid $\mu \eqdef \inf\{b \in [0,\, \lambda] \mid \varphi_{H}(b) \geq 1\}$; this threshold bid $\mu \in [0,\, \lambda]$ is well defined regardless of our non-monopoly assumption for bidder $H$. {\bf \Cref{fact:BNE_independent:monopolist:A}} and {\bf \Cref{fact:BNE_independent:monopolist:B}} will be helpful for the later proof; only {\bf \Cref{fact:BNE_independent:monopolist:B}} relies on our non-monopoly assumption for bidder $H$.

\setcounter{fact}{0}

\begin{fact}
\label{fact:BNE_independent:monopolist:A}
{\bf (I)}~Bid distribution $B_{H}$ has no density on the interval $b \in (0,\, \mu)$, namely $B_{H}(\mu) = B_{H}(0)$, and has densities almost everywhere on the interval $b \in (\mu,\, \lambda)$. \\
{\bf (II)}~The bid-to-value mapping $\varphi_{H}(b) \leq 1$ for $b \in [0,\, \lambda]$; the equality holds when $b \in [\mu,\, \lambda]$.
\end{fact}

\begin{proof}
By the definition of $\mu = \inf\{b \in [0,\, \lambda] \mid \varphi_{H}(b) \geq 1\}$, a nonzero bid $s_{H}(v_{H}) \in (0,\, \lambda]$ due to the nonzero value $\{v_{H} = 1\}$ can be further restricted to the range $s_{H}(v_{H}) \in (\mu,\, \lambda]$. (Recall \Cref{lem:bid_distribution:continuity} of \Cref{lem:bid_distribution} that the bid CDF $B_{H}(b)$ is a {\em continuous} function over the bid support $b \in [0,\, \lambda]$.) Namely, bid distribution $B_{H}$ has no density on the interval $b \in (0,\, \mu)$ and thus $B_{H}(\mu) = B_{H}(0)$.

For the sake of contradiction, assume that bid distribution $B_{H}$ has no density around some bid $\beta \in (\mu,\, \lambda)$, namely $B'_{H}(\beta) = 0$. Then at this particular bid $\beta \in (\mu,\, \lambda)$, the two bid-to-value mappings $\varphi_{H}(b)$ and $\varphi_{L}(b)$ satisfy that
\begin{align*}
    1 & ~\leq~ \varphi_{H}(\beta) ~=~ \beta + \Big(n \cdot B'_{L}(\beta) / B_{L}(\beta)\Big)^{-1}
    && \text{$\beta \in (\mu,\, \lambda)$} \\
    & ~\leq~ \varphi_{L}(\beta)
    ~=~ \beta + \Big((n - 1) \cdot B'_{L}(\beta) / B_{L}(\beta) + B'_{H}(\beta) / B_{H}(\beta)\Big)^{-1}.
    && \text{$B'_{H}(\beta) = 0$}
\end{align*}
This means value distribution $V_{L}$ has densities around some value $v_{\beta} \geq 1$. Precisely, value distribution $V_{L}$ can be reconstructed via the parametric equation $\{(v,\, V_{L}) = (\varphi_{L}(b),\, B_{L}(b)) \mid b \in [0,\, \lambda]\}$. Furthermore, bid distribution $B_{L}$ has densities almost everywhere over the bid support $b \in [0,\, \lambda]$ (\Cref{lem:BNE_independent:L}), including the particular bid $\beta \in (\mu,\, \lambda)$ for which $v_{\beta} = \varphi_{L}(\beta) \geq 1$. But this contradicts our construction -- The hypothetical value $v_{\beta} = \varphi_{L}(\beta) \geq 1$ is bounded away from value distribution $V_{L}$'s support $\supp(V_{L}) = [0,\, 1 - \frac{2n - 2}{2n - 1} \cdot 2 / e^{2}]$ (\Cref{exp:BNE_independent}).

Refuting the above assumption results in {\bf Part~(I)}: Bid distribution $B_{H}$ has densities almost everywhere on $b \in (\mu,\, \lambda)$. All those densities stem from the value $\{v_{H} = 1\}$, since bidder $H$ has a Bernoulli random value $v_{H} \in \{0,\, 1\}$ and (by assumption) is a non-monopoly bidder. Therefore, we have $\varphi_{H}(b) = 1$ for $b \in [\mu,\, \lambda]$. This together with monotonicity of the bid-to-value mapping $\varphi_{H}(b)$, immediately gives {\bf Part~(II)}. This finishes the proof.
\end{proof}

\begin{fact}
\label{fact:BNE_independent:monopolist:B}
Assume that bidder $H$ is a non-monopoly bidder $B_{H}(0) = V_{H}(0) = \epsilon$. \\
{\bf (I)}~The supremum bid $\lambda \leq \lambda^{*} = 1 - 4 / e^{2}$. \\
{\bf (II)}~$B_{L}(b) \geq B_{L}^{*}(b)$ for $b \in [0,\, \lambda]$. \\
{\bf (III)}~$B'_{L}(b) / B_{L}(b) \geq {B_{L}^{*}}'(b) / B_{L}^{*}(b)$ for $b \in [0,\, \lambda]$; the equality holds when $b \in [\mu,\, \lambda]$. \\
{\bf (IV)}~$B'_{H}(b) / B_{H}(b) \leq {B_{H}^{*}}'(b) / B_{H}^{*}(b)$ for $b \in [\mu,\, \lambda]$.
\end{fact}

\begin{proof}
The bid-to-value mapping $\varphi_{H}(b) = b + \frac{1}{n} \cdot B_{L}(b) / B'_{L}(b) \leq 1$ over the bid support $b \in [0,\, \lambda]$; the equality holds when $b \in [\mu,\, \lambda]$.
In contrast, the focal bid-to-value mapping $\varphi_{H}^{*}(b) = b + \frac{1}{n} \cdot B_{L}^{*}(b) / {B_{L}^{*}}'(b) = 1$ over the focal bid support $b \in [0,\, \lambda^{*}]$.
Therefore, for $b \in [0,\, \min(\lambda,\, \lambda^{*})]$ we have $B'_{L}(b) / B_{L}(b) \geq {B_{L}^{*}}'(b) / B_{L}^{*}(b)$ and thus
\[
    B_{L}(b) / B_{L}(0)
    ~=~ \exp\Big(\int_{0}^{b} B'_{L}(b) / B_{L}(b) \cdot \d x\Big)
    ~\geq~ B_{L}^{*}(b) / B_{L}^{*}(0)
    ~=~ \exp\Big(\int_{0}^{b} {B_{L}^{*}}'(b) / B_{L}^{*}(b) \cdot \d x\Big).
\]
The two bid distributions have the same probability mass $B_{L}(0) = B_{L}^{*}(0) = V_{L}(0) = \sqrt[n]{4 / e^{2}}$ at the zero bid, so we have $B_{L}(b) \geq B_{L}^{*}(b)$ for $b \in [0,\, \min(\lambda,\, \lambda^{*})]$. To achieve the boundary conditions $B_{L}(\lambda) = 1$ and $B_{L}^{*}(\lambda^{*}) = 1$ at the respective supremum bids $\lambda$ and $\lambda^{*}$, we must have {\bf Part~(I)} that $\lambda \leq \lambda^{*} = 1 - 4 / e^{2}$.

{\bf Part~(II)} and {\bf Part~(III)}, including the equality $B'_{L}(b) / B_{L}(b) = {B_{L}^{*}}'(b) / B_{L}^{*}(b)$ for $b \in [\mu,\, \lambda]$, can be easily inferred from the above arguments.

\begin{figure}[t]
    \centering
    \includegraphics[width = .9\textwidth]{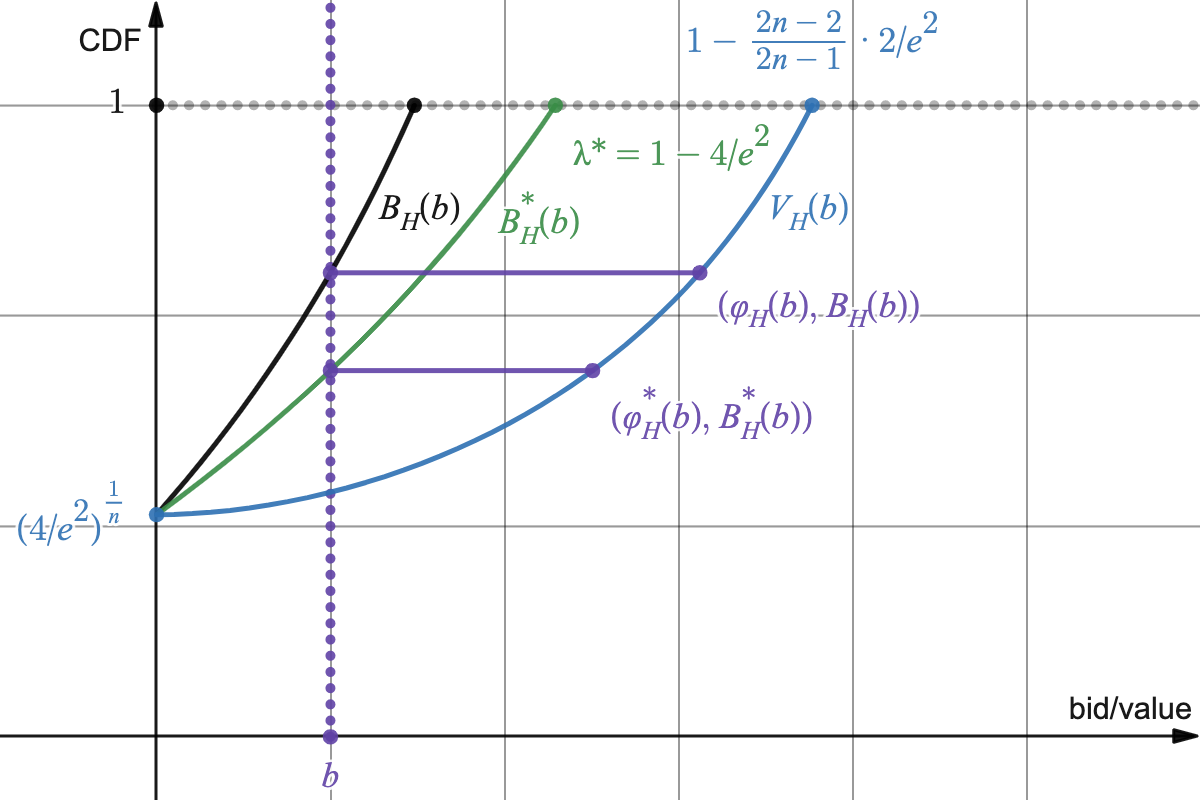}
    \caption{Demonstration for the proof of {\bf Part~(IV)} of {\bf \Cref{fact:BNE_independent:monopolist:B}}.}
    \label{fig:BNE_independent:proof}
\end{figure}

Value distribution $V_{L}$ can be reconstructed from EITHER bid distributions $B_{H} \otimes \{B_{L}\}^{\otimes n}$ OR the focal bid distributions $B_{H}^{*} \otimes \{B_{L}^{*}\}^{\otimes n}$, via the parametric equations $\{(v,\, V_{L}) = (\varphi_{L}(b),\, B_{L}(b)) \mid b \in [0,\, \lambda]\}$ or $\{(v,\, V_{L}) = (\varphi_{L}^{*}(b),\, B_{L}^{*}(b)) \mid b \in [0,\, \lambda^{*}]\}$.
As \Cref{fig:BNE_independent:proof} suggests, this observation together with {\bf Part~(II)} implies that $\varphi_{L}(b) \geq \varphi_{L}^{*}(b)$ for $b \in [0,\, \lambda]$.
Especially, on the restricted interval $b \in [\mu,\, \lambda]$, we can deduce that
\begin{align*}
    \varphi_{L}(b)
    & ~=~ b + \Big((n - 1) \cdot B'_{L}(b) / B_{L}(b) + B'_{H}(b) / B_{H}(b)\Big)^{-1} \\
    & ~\geq~ \varphi_{L}^{*}(b)
    ~=~ b + \Big((n - 1) \cdot {B_{L}^{*}}'(b) / B_{L}^{*}(b) + {B_{H}^{*}}'(b) / B_{H}^{*}(b)\Big)^{-1}.
\end{align*}
Rearranging this equation and applying {\bf Part~(III)} of {\bf \Cref{fact:BNE_independent:monopolist:B}} (that $B'_{L}(b) / B_{L}(b) = {B_{L}^{*}}'(b) / B_{L}^{*}(b)$ when $b \in [\mu,\, \lambda]$), we can conclude {\bf Part~(IV)} immediately. This finishes the proof of {\bf \Cref{fact:BNE_independent:monopolist:B}}.
\end{proof}

However, combining everything together, we can derive the following contradiction:
\begin{align*}
    1 ~=~ B_{H}(\lambda)
    & ~=~ B_{H}(\lambda) / B_{H}(\mu) \cdot B_{H}(\mu) \phantom{\Big.} \\
    & ~=~ \exp\Big(\int_{\mu}^{\lambda} B'_{H}(b) / B_{H}(b) \cdot \d x\Big) \cdot B_{H}(\mu) \\
    & ~=~ \exp\Big(\int_{\mu}^{\lambda} B'_{H}(b) / B_{H}(b) \cdot \d x\Big) \cdot \epsilon
    && \text{{\bf Part~(I)} of {\bf \Cref{fact:BNE_independent:monopolist:A}}} \\
    & ~\leq~ \exp\Big(\int_{\mu}^{\lambda} {B_{H}^{*}}'(b) / B_{H}^{*}(b) \cdot \d x\Big) \cdot \epsilon
    && \text{{\bf Part~(IV)} of {\bf \Cref{fact:BNE_independent:monopolist:B}}} \\
    & ~=~ B_{H}^{*}(\lambda) / B_{H}^{*}(\mu) \cdot \epsilon \phantom{\Big.} \\
    & ~\leq~ 4\epsilon ~<~ 1 / 2. \phantom{\Big.}
\end{align*}
Here the last line uses $B_{H}^{*}(\mu) \geq B_{H}^{*}(0) = 1 / 4$, $B_{H}^{*}(\lambda) \leq B_{H}^{*}(\lambda^{*}) = 1$, and $\epsilon \in (0,\, 1 / 8)$; all of which can be found from \Cref{exp:BNE_independent} and {\bf Part~(I)} of {\bf \Cref{fact:BNE_independent:monopolist:B}}.

Refute our assumption: Bidder $H$ is the unique monopolist $B_{H}(0) > V_{H}(0) = \epsilon$; the probability mass $B_{H}(0) > \epsilon$ at the zero bid stems from BOTH a zero value $\{v_{H} = 0\}$ (by the whole amount $\epsilon = \Pr[v_{H} = 0]$) AND a nonzero value $\{v_{H} = 1\}$ (by a partial amount $B_{H}(0) - \epsilon \leq \Pr[v_{H} = 1]$).
This implies that the threshold bid $\mu = \inf\{b \in [0,\, \lambda] \mid \varphi_{H}(b) \geq 1\}$ is zero $\mu = 0$. As a consequence, we can infer \Cref{lem:BNE_independent:H} from {\bf \Cref{fact:BNE_independent:monopolist:A}}. This finishes the proof.
\end{proof}

\begin{lemma}[Uniqueness of Equilibria]
\label{lem:BNE_independent:unique}
The following hold:
\begin{enumerate}[font = {\em\bfseries}]
    \item\label{lem:BNE_independent:unique:1}
    The supremum bid $\lambda = 1 - 4 / e^{2}$, the same as the focal supremum bid $\lambda = \lambda^{*}$.
    
    \item\label{lem:BNE_independent:unique:2}
    Bid distributions $\{B_{L,\, i}\}_{i \in [n]} = \{B_{L}\}^{\otimes n}$ are given by $B_{L}(b) = \sqrt[n]{(1 - \lambda) / (1 - b)}$ for $b \in [0,\, \lambda]$, the same as the focal bid distribution $B_{L} \equiv B_{L}^{*}$.
    
    \item\label{lem:BNE_independent:unique:3}
    Bid distribution $B_{H}$ is given by the implicit equation $b = 1 - 4B_{H} \cdot e^{2 - 4\sqrt{B_{H}}}$ for $B_{H} \in [1 / 4,\, 1]$, the same as the focal bid distribution $B_{H} \equiv B_{H}^{*}$.
\end{enumerate}
\end{lemma}


\begin{proof}
Following \Cref{lem:BNE_independent:H}, over the bid support $b \in [0,\, \lambda]$, bidder $H$ has a constant bid-to-value mapping $\varphi_{H}(b) = b + \frac{1}{n} \cdot B_{L}(b) / B'_{L}(b) = 1$. By resolving this ODE, under the boundary condition $B_{L}(0) = \sqrt[n]{4 / e^{2}}$ at the infimum bid $= 0$, we have $B_{L}(b) = \sqrt[n]{(4 / e^{2}) / (1 - b)}$ for $b \in [0,\, \lambda]$. Plugging this CDF formula into the other boundary condition $B_{L}(\lambda) = 1$ at the supremum bid $= \lambda$, we can deduce that $\lambda = 1 - 4 / e^{2}$. {\bf \Cref{lem:BNE_independent:unique:1}} and {\bf \Cref{lem:BNE_independent:unique:2}} follow then.

It remains to show {\bf \Cref{lem:BNE_independent:unique:3}}. By construction (\Cref{exp:BNE_independent}), bidders $\{L_{i}\}_{i \in [n]}$ have identical value distributions $\{V_{L}\}^{\otimes n}$ given by $V_{L}(1 - \frac{n - (t - 1)}{n t - (t - 1)} \cdot t^{2} \cdot e^{2 - 2t}) = \sqrt[n]{4 / t^{2} \cdot e^{2t - 4}}$ for $t \in [1,\, 2]$.
Those value distributions also can be reconstructed through the parametric equation $\big\{(v,\, V_{L}) = (\varphi_{L}(b),\, B_{L}(b)) \bigmid b \in [0,\, \lambda]\big\}$. As a combination, in terms of $t \in [1,\, 2]$, the bid-to-value mapping $\varphi_{L}$ can be rewritten as follows:
\begin{equation}
\label{eq:BNE_independent:unique:1}
    \varphi_{L} = 1 - \frac{n - (t - 1)}{n t - (t - 1)} \cdot t^{2} \cdot e^{2 - 2t}.
\end{equation}
Similarly, we can deduce that $\sqrt[n]{4 / t^{2} \cdot e^{2t - 4}} = B_{L}(b) = \sqrt[n]{(1 - \lambda) / (1 - b)}$ and thus rewrite the bid $b = 1 - t^{2} \cdot e^{2 - 2t}$ for $t \in [1,\, 2]$. Then, the derivative $\d b / \d t = (2t^{2} - 2t) \cdot e^{2 - 2t}$.

On the other hand, each bidder $L_{i}$ for $i \in [n]$ competes with bidders $\{H\} \cup \{L_{j}\}_{j \in [n] \setminus \{i\}}$, hence the competing bid distribution $\calB_{-L}(b) = B_{H}(b) \cdot \big(B_{L}(b)\big)^{n - 1}$. Consequently, in terms of $t \in [1,\, 2]$, the bid-to-value mapping $\varphi_{L} = b + \frac{\calB_{-L}}{\d \calB_{-L} / \d b}$ also can be rewritten as follows:
\begin{align}
    \varphi_{L}
    & ~=~ b + \Big((n - 1) \cdot \frac{\d B_{L} / \d b}{B_{L}} + \frac{\d B_{H} / \d b}{B_{H}}\Big)^{-1}
    \nonumber \\
    & ~=~ b + \Big(\frac{n - 1}{n} \cdot \frac{1}{1 - b} + \frac{\d B_{H} / \d t}{B_{H}} \cdot \frac{1}{(2t^{2} - 2t) \cdot e^{2 - 2t}}\Big)^{-1}
    \nonumber \\
    & ~=~ (1 - t^{2} \cdot e^{2 - 2t}) + \Big(\frac{n - 1}{n} \cdot \frac{1}{t^{2} \cdot e^{2 - 2t}} + \frac{\d B_{H} / \d t}{B_{H}} \cdot \frac{1}{(2t^{2} - 2t) \cdot e^{2 - 2t}}\Big)^{-1}.
    \label{eq:BNE_independent:unique:2}
\end{align}
Here the second line applies $B_{L}(b) = \sqrt[n]{(1 - \lambda) / (1 - b)}$ and $\d b / \d t = (2t^{2} - 2t) \cdot e^{2 - 2t}$, and the last line applies $b = 1 - t^{2} \cdot e^{2 - 2t}$.

The above two formulas for the $\varphi_{L}$ must be identical for $t \in [1,\, 2]$, so we can deduce that
\begin{align}
    & \text{\Cref{eq:BNE_independent:unique:1}}
    ~=~ \text{\Cref{eq:BNE_independent:unique:2}} \phantom{\Big.}
    \nonumber \\
    \iff~~ & \Big(1 - \frac{n - (t - 1)}{n t - (t - 1)}\Big) \cdot t^{2} \cdot e^{2 - 2t}
    ~=~ \Big(\frac{n - 1}{n} \cdot \frac{1}{t^{2} \cdot e^{2 - 2t}} + \frac{\d B_{H} / \d t}{B_{H}} \cdot \frac{1}{(2t^{2} - 2t) \cdot e^{2 - 2t}}\Big)^{-1}
    \nonumber \\
    \iff~~ & \frac{n t - (t - 1)}{n \cdot (t - 1)} \cdot \frac{1}{t^{2} \cdot e^{2 - 2t}}
    ~=~ \frac{n - 1}{n} \cdot \frac{1}{t^{2} \cdot e^{2 - 2t}} + \frac{\d B_{H} / \d t}{B_{H}} \cdot \frac{1}{(2t^{2} - 2t) \cdot e^{2 - 2t}}
    \nonumber \\
    \iff~~ & \frac{1}{t - 1} \cdot \frac{1}{t^{2} \cdot e^{2 - 2t}}
    ~=~ \frac{\d B_{H} / \d t}{B_{H}} \cdot \frac{1}{(2t^{2} - 2t) \cdot e^{2 - 2t}}
    \nonumber \\
    \iff~~ & \frac{2}{t} ~=~ \frac{\d B_{H} / \d t}{B_{H}}
    \nonumber \\
    \iff~~ & \frac{\d}{\d t}(2\ln t) ~=~ \frac{\d}{\d t}(\ln B_{H}).
    \label{eq:BNE_independent:unique3}
\end{align}
Especially, when $t = 2$, the bid $b = 1 - t^{2} \cdot e^{2 - 2t} = 1 - 4 / e^{2}$ achieves the supremum bid $\lambda = 1 - 4 / e^{2}$ and we have the boundary condition $B_{H}(b) = B_{H}(\lambda) = 1$. Resolving ODE~\eqref{eq:BNE_independent:unique3} under this boundary condition, we derive that $2\ln(t / 2) = \ln B_{H}$ and thus $B_{H} = t^{2} / 4$. Plugging this into the formula $b = 1 - t^{2} \cdot e^{2 - 2t}$ for $t \in [1,\, 2]$ gives the implicit equation $b = 1 - 4B_{H} \cdot e^{2 - 4\sqrt{B_{H}}}$ for $B_{H} \in [1 / 4,\, 1]$, as desired.
This finishes the proof of \Cref{lem:BNE_independent:unique}.
\end{proof}

\section{Bayesian Nash Equilibria for Correlated Valuations}
\label{sec:BNE_correlated}

In this section, we give the tight {\PoS} bound for correlated valuations. 

\begin{theorem}[Tight {\PoS}]
\label{thm:BNE_correlated}
Regarding {\BayesNashEquilibria} for correlated valuations, 
the {\PriceofStability} is $1 - 1 / e \approx 0.6321$.
\end{theorem}

The proof framework is similar to that in the previous section, given that this tight {\PoS} bound $= 1 - 1 / e$ also coincides with the tight {\PoA} bound by \cite{ST13,S14}. We slightly modify the tight {\PoA} instance from \cite[Appendix~A.2]{S14} into the next instance and prove that it has one unique equilibrium.
(Namely, by setting $\epsilon = 0$, \Cref{exp:BNE_correlated} is precisely the original tight {\PoA} instance.)

\begin{example}
\label{exp:BNE_correlated}
Given an arbitrarily small constant $\epsilon \in (0,\, 1)$, consider the instance $\{H,\, L_{1},\, L_{2}\}$ in terms of the joint value distribution $\bv = (v_{H},\, v_{L,\, 1},\, v_{L,\, 2}) \sim \bV$.
\begin{itemize}
    \item Bidder $H$ has an {\em independent} Bernoulli random value $\Pr[v_{H} = 0] = \epsilon$ and $\Pr[v_{H} = 1] = 1 - \epsilon$. Denote by $V_{H}$ this (marginal) value distribution.
    
    \item Bidders $L_{1}$ and $L_{2}$ have {\em perfectly correlated} and {\em identical} values $v_{L,\, 1} \equiv v_{L,\, 2}$, which follow the (marginal) value distribution $V_{L}(v) = \frac{\epsilon + 1 / e}{\epsilon + 1 - v}$ for $v \in [0,\, 1 - 1 / e]$.
\end{itemize}
The considered {\FirstPriceAuction} $\alloc \in \bbFPA$, under the all-zero bid profile $\bb = (b_{H},\, b_{L,\, 1},\, b_{L,\, 2}) = \zeros$, favors bidder $H$, but otherwise is arbitrary.
\end{example}

\subsection{The focal equilibrium}
\label{subsec:BNE_correlated:focal}

The {\em \textbf{focal} strategy profile} $\bs^{*} = \{s_{H}^{*}\} \otimes \{s_{L}^{*}\}^{\otimes 2}$ is given as follows.
\begin{itemize}
    \item Bidder $H$ has a {\em fixed} strategy $s_{H}^{*}(v_{H}) \equiv 0$ for $v_{H} \in \{0,\, 1\}$.
    
    \item Bidders $L_{1}$ and $L_{2}$ have identical and {\em truthful} strategies $s_{L}^{*}(v) \equiv v$ for $v \in [0,\, 1 - 1 / e]$.
\end{itemize}

\begin{lemma}[Equilibrium]
\label{lem:BNE_correlated:focal}
The focal strategy profile $\bs^{*} = s_{H}^{*} \otimes \{s_{L}^{*}\}^{\otimes 2}$ forms a {\BayesNashEquilibrium} $\bs^{*} \in \bbBNE(\bV)$.
\end{lemma}

\begin{proof}
Bidder $L_{1}$ satisfies the equilibrium condition: Given a specific bid $s_{L}^{*}(v_{1}) = b \in [0,\, 1 - 1 / e]$, bidder $L_{1}$' value and bidder $L_{2}$'s value/bid ALL must be the same $v_{1} = v_{2} = s_{L}^{*}(v_{1}) = s_{L}^{*}(v_{2}) = b$. Clearly, bidder $L_{1}$ gains a zero utility $= 0$ and cannot benefit from a deviation bid $b' \neq b$, namely a higher bid $b' > b$ gives a nonpositive utility $\leq 0$ and a lower bid $b' < b = s_{L}^{*}(v_{2})$ makes bidder $L_{1}$ lose to bidder $L_{2}$. By symmetry, bidder $L_{2}$ also meets the equilibrium condition.

Bidder $H$'s competing bid distribution $\max(s_{L}^{*}(v_{1}),\, s_{L}^{*}(v_{2})) \sim \calB_{-H}^{*}$ is exactly bidders $L_{1}$ and $L_{2}$'s value distribution $V_{L}$. For any value $v_{H} \in \{0,\, 1\}$ and any bid $b \in [0,\, 1 - 1 / e]$, bidder $H$ gains an interim utility $= (v_{H} - b) \cdot V_{L}(b) = \frac{v_{H} - b}{\epsilon + 1 - b} \cdot (\epsilon + 1 / e)$. We can easily verify that, under the either value $v_{H} \in \{0,\, 1\}$, the zero bid $b = 0$ is the {\em unique} maximizer for this interim utility formula; thus bidder $H$ also satisfies the equilibrium condition. This finishes the proof.
\end{proof}

\begin{lemma}[Efficiency]
\label{lem:BNE_correlated:efficiency}
The following hold for the focal equilibrium $\bs^{*} = \{s_{H}^{*}\} \otimes \{s_{L}^{*}\}^{\otimes 2}$:
\begin{enumerate}[font = {\em\bfseries}]
    \item\label{lem:BNE_correlated:efficiency:1}
    The expected optimal {\SocialWelfare} $\OPT(\bV,\, \bs^{*}) \geq 1 - \epsilon$.

    \item\label{lem:BNE_correlated:efficiency:2}
    The expected auction {\SocialWelfare} $\FPA(\bV,\, \bs^{*}) \leq 1 - 1 / e$.
\end{enumerate}
\end{lemma}

\begin{proof}
Let us prove {\bf \Cref{lem:BNE_correlated:efficiency:1,lem:BNE_correlated:efficiency:2}} one by one.

\vspace{.1in}
\noindent
{\bf \Cref{lem:BNE_correlated:efficiency:1}.}
The realized optimal {\SocialWelfare} is at least bidder $H$'s realized value $v_{H} \sim V_{H}$, namely a Bernoulli random value $\Pr [v_{H} = 0] = \epsilon$ and $\Pr [v_{H} = 1] = 1 - \epsilon$ (\Cref{exp:BNE_correlated}). In expectation, we have $\OPT(\bV,\, \bs^{*}) \geq \E[v_{H}] = 1 - \epsilon$. {\bf \Cref{lem:BNE_correlated:efficiency:1}} follows then.

\vspace{.1in}
\noindent
{\bf \Cref{lem:BNE_correlated:efficiency:2}.}
The first-order bid distribution $\max(\bs^{*}(\bv)) \sim \calB^{*}$ is exactly bidders $L_{1}$ and $L_{2}$'s value distribution $V_{L}$.
Conditioned on a zero first-order bid $\{\max(\bs^{*}(\bv)) = 0\}$, bidder $H$ wins and the realized auction {\SocialWelfare} is her value $v_{H} \sim V_{H}$. And conditioned on a nonzero first-order bid $\{\max(\bs^{*}(\bv)) > 0\}$, either bidder $L_{1}$ or bidder $L_{2}$ wins and the realized auction {\SocialWelfare} is their identical value $v_{1} = v_{2} = s_{L}^{*}(v_{1}) = s_{L}^{*}(v_{2}) = \max(\bs^{*}(\bv)) > 0$. In expectation, we have
\begin{align*}
    \FPA(\bV,\, \bs^{*})
    & ~=~ \E[V_{H}] \cdot V_{L}(0) + \E[V_{L}] \phantom{\Big.} \\
    & ~=~ (1 - \epsilon) \cdot \frac{\epsilon + 1 / e}{\epsilon + 1} + \int_{0}^{1 - 1 / e} \big(1 - V_{L}(v)\big) \cdot \d v \phantom{\Big.} \\
    & ~=~ 1 - 1 / e - (\epsilon + 1 / e) \cdot \Big(\frac{2\epsilon}{1 + \epsilon} - \ln\Big(\frac{e \epsilon + 1}{\epsilon + 1}\Big)\Big) \phantom{\Big.} \\
    & ~\leq~ 1 - 1 / e. \phantom{\Big.}
\end{align*}
Here the last two steps can be easily verified via elementary algebra. This finishes the proof.
\end{proof}

\subsection{Uniqueness of equilibria}
\label{subsec:BNE_correlated:unique}

In this part, we prove the focal equilibrium $\bs^{*}$ is the unique equilibrium for \Cref{exp:BNE_correlated}. Once again, we start with a generic equilibrium $\bs = \{s_{H},\, s_{L,\, 1},\, s_{L,\, 2}\} \in \bbBNE(\bV)$.

\begin{lemma}
\label{lem:BNE_correlated:no_overbid}
Each bidder $i \in \{H,\, L_{1},\, L_{2}\}$ cannot overbid, $s_{i}(v) \leq v$ almost surely over the randomness of the strategy $s_{i}$, everywhere $v \in \supp(V_{i})$ except on a zero-measure set of values.
\end{lemma}

\begin{proof}
First, on having a zero value $\{v_{i} = 0\}$, each bidder $i \in \{H,\, L_{1},\, L_{2}\}$ has a zero bid $s_{i}(v_{i}) = 0$ almost surely. Otherwise, with a nonzero probability $> 0$, the following event occurs. \\
$\{\bv = \zeros \wedge \bs(\bv) \neq \zeros\}$: The value profile $\bv$ is all-zero but the bid profile $\bs(\bv)$ is not. \\
But conditioned on this, the allocated bidder $\alloc(\bs(\bv))$ realizes a negative utility $= -\max(\bs(\bv)) < 0$, which contradicts the equilibrium condition.

Second, on having a nonzero value $\{v_{H} = 1\}$, bidder $H$ cannot overbid $s_{H}(v_{H}) \leq v_{H}$. (Recall that value $v_{H} \sim V_{H}$ is independent from the other two values $v_{L,\, 1} \equiv v_{L,\, 2} \sim V_{L}$.) Otherwise, with a nonzero probability $> 0$, the next two events occur simultaneously. \\
$\{v_{H} = 1 \wedge s_{H}(v_{H}) > 1\}$: Bidder $H$ has a nonzero value and overbids; \\
$\{v_{L,\, 1} = v_{L,\, 2} = s_{L,\, 1}(v_{L,\, 1}) = s_{L,\, 1}(v_{L,\, 1}) = 0\}$: Bidders $L_{1}$ and $L_{2}$ have zero values and zero bids. \\
But conditioned on this, bidder $H$ gets allocated and realizes a negative utility $= v_{H} - s_{H}(v_{H}) < 0$, which contradicts the equilibrium condition.

Third, on having nonzero values $(v_{L,\, 1} \equiv v_{L,\, 2}) > 0$, bidders $L_{1}$ and $L_{2}$ cannot overbid. Otherwise, with a nonzero probability $> 0$, the next two events occur simultaneously. \\
$\{\max(s_{L,\, 1}(v_{L,\, 1}),\, s_{L,\, 2}(v_{L,\, 2})) > (v_{L,\, 1} \equiv v_{L,\, 2}) > 0\}$: Bidders $L_{1}$ and $L_{2}$ have identical and nonzero values $ > 0$ and at least one of them overbids; \\
$\{v_{H} = s_{H}(v_{H}) = 0\}$: Bidder $H$ has a zero value and a zero bid. \\
But conditioned on this, the allocated bidder $\alloc(\bs(\bv)) \in \argmax(\bs(\bv))$ is the higher bidder between $L_{1}$ and $L_{2}$, realizing a negative utility $< 0$, which contradicts the equilibrium condition.

This finishes the proof of \Cref{lem:BNE_correlated:no_overbid}.
\end{proof}

\begin{lemma}
\label{lem:BNE_correlated:truthful}
Bidders $L_{1}$ and $L_{2}$ play the truthful strategies, $s_{L,\, 1}(v) = s_{L,\, 2}(v) = v$ almost surely, everywhere $v \in [0,\, 1 - 1 / e]$ except on a zero-measure set of values.
\end{lemma}

\begin{proof}
Following \Cref{lem:BNE_correlated:no_overbid}, bidders $L_{1}$ and $L_{2}$ cannot overbid $s_{L,\, 1}(v),\, s_{L,\, 2}(v) \ngtr v$ and play the truthful strategies $s_{L,\, 1}(0) = v_{L,\, 2}(0) = 0$ on having the zero values. It remains to show that bidders $L_{1}$ and $L_{2}$ also cannot ``shade'' their bids, namely $s_{L,\, 1}(v),\, s_{L,\, 2}(v) \nless v$ for $v > 0$.

Assume the opposite: For some nonzero value $(v_{L,\, 1} \equiv v_{L,\, 2}) = v > 0$, either or both of $\{L_{1},\, L_{2}\}$ ``shades'' her bid with a nonzero probability $\Pr_{s_{L,\, 1},\, s_{L,\, 2}}[\min(s_{L,\, 1}(v),\, s_{L,\, 2}(v)) < v] > 0$. Let us do case analysis conditioned on the event $\calE = \{(v_{L,\, 1} \equiv v_{L,\, 2}) = v > 0\}$:
\begin{itemize}
    \item {\bf Case~(I).}
    Exactly one bidder $\in \{L_{1},\, L_{2}\}$ ``shades'' her bid with a nonzero probability.
    
    Without loss of generality, bidder $L_{1}$ plays the shade strategy $\Pr_{s_{L,\, 1}}[s_{L,\, 1}(v) < v] > 0$, while bidder $L_{2}$ plays the truthful strategy $\Pr_{s_{L,\, 2}}[s_{L,\, 2}(v) = v] = 1$. But if so, bidder $L_{2}$ can benefit from a certain deviation bid $b^{*}$ against the current zero utility $= 0$ from the {\em truthful} strategy. Specifically, bidder $L_{1}$'s infimum strategy $\underline{s}_{1} = \inf(\supp(s_{L,\, 1}(v)))$ is bounded away from the considered value $\underline{s}_{1} < v$. Using the deviation bid $b^{*} = \frac{1}{2}(\underline{s}_{1} + v)$, bidder $L_{2}$ realizes a positive utility $= v - b^{*} > 0$ on winning, and wins with a nonzero probability $> 0$: Independently, \\
    (i)~bidder $H$ loses with probability $\geq \Pr_{v_{H},\, s_{H}}[s_{H}(v_{H}) < b^{*}] \geq \Pr_{v_{H},\, s_{H}}[v_{H} = 0] = \epsilon$, because bidder $H$ on having a zero value $\{v_{H} = 0\}$ also has a zero bid $s_{H}(v_{H}) = 0$ (\Cref{lem:BNE_correlated:no_overbid}); \\
    (ii)~bidder $L_{1}$ loses with a nonzero probability $\geq \Pr_{s_{L,\, 1}}[s_{L,\, 1}(v) < b^{*}] > 0$, as a consequence of $b^{*} = \frac{1}{2}(\underline{s}_{1} + v) > \underline{s}_{1} = \inf(\supp(s_{L,\, 1}(v)))$.
    
    Thus, bidder $L_{2}$ can benefit from a certain deviation bid $b^{*}$. This contradicts the equilibrium condition and refutes {\bf Case~(I)}.
    
    \item {\bf Case~(II).}
    Each bidder $\in \{L_{1},\, L_{2}\}$ ``shades'' her bid with a nonzero probability.
    
    Consider the infimum strategies $\underline{s}_{1} = \inf(\supp(s_{L,\, 1}(v))) < v$ and $\underline{s}_{2} = \inf(\supp(s_{L,\, 2}(v))) < v$ and their likelihoods $p_{1} = \Pr_{s_{L,\, 1}}[s_{L,\, 1}(v) = \underline{s}_{1}]$ and $p_{2} = \Pr_{s_{L,\, 2}}[s_{L,\, 2}(v) = \underline{s}_{2}]$.
    \begin{itemize}
        \item {\bf Case~(a).}
        $\underline{s}_{1} = \underline{s}_{2} = \underline{s}$ for some bid $\underline{s} \in [0,\, v]$ and $p_{1},\, p_{2} > 0$.
        
        A tiebreak at the bid $\underline{s}$ occurs with a nonzero probability $\Pr_{v_{H},\, \bs}[\max(\bs(\bv)) = \underline{s} \mid \calE] = \Pr_{v_{H},\, s_{H}}[s_{H}(v_{H}) \leq \underline{s}] \cdot p_{1} \cdot p_{2} > 0$, since bidder $H$ on having a zero value $\Pr[v_{H} = 0] = \epsilon$ also has a zero bid $s_{H}(v_{H}) = 0$ (\Cref{lem:BNE_correlated:no_overbid}). In this tiebreak $\calE \wedge \{\max(\bs(\bv)) = \underline{s}\}$, at least one bidder between $L_{1}$ and $L_{2}$ loses with a nonzero probability $> 0$, say bidder $L_{1}$. But this means, using a higher but close enough deviation bid $b^{*} \searrow \underline{s}$, bidder $L_{1}$ realizes an arbitrarily close positive utility $= (v - b^{*}) \nearrow (v - \underline{s}) > 0$ on winning, yet the winning probability increases by a nonzero amount $\geq \Pr_{v_{H},\, \bs}[\max(\bs(\bv)) = \underline{s} \mid \calE] > 0$.
        
        Thus, bidder $L_{1}$ can benefit from a higher but close enough deviation bid $b^{*} \searrow \underline{s}$. This contradicts the equilibrium condition and refutes {\bf Case~(a)}.
        
        \item {\bf Case~(b).}
        Either $\underline{s}_{1} \neq \underline{s}_{2}$ or $p_{1} \cdot p_{2} = 0$.
        
        If $\underline{s}_{1} \neq \underline{s}_{2}$, without loss of generality we have $\underline{s}_{1} < \underline{s}_{2} < v$. But if so, bidder $L_{1}$ gains a nonzero utility $> 0$ from any bid $b^{*} \in (\underline{s}_{2},\, v)$, in contrast to a zero utility $= 0$ from any bid $\in [\underline{s}_{1},\, \underline{s}_{2})$. The current strategy $s_{L,\, 1}(v)$ has densities on the ``useless'' bids $\in [\underline{s}_{1},\, \underline{s}_{2})$ and cannot be utility-optimal.
        
        If $\underline{s}_{1} = \underline{s}_{2} = \underline{s}$ for some bid $\underline{s} \in [0,\, v]$ and $p_{1} \cdot p_{2} = 0$, without loss of generality we have $p_{1} = 0$. But if so, bidder $L_{2}$ gains a nonzero utility $> 0$ from any bid $b^{*} \in [\underline{s},\, v]$ that is bounded away from both $\underline{s}$ and $v$, in contrast to an arbitrarily small utility $\searrow 0$ (as the winning probability $\searrow 0$) when her bid $\in [\underline{s},\, v]$ approaches the infimum strategy $\searrow \underline{s}$. The current strategy $s_{L,\, 2}(v)$ has densities in any neighborhood around the infimum strategy $\underline{s}$ and cannot be utility-optimal.

        Thus, at least one bidder between $L_{1}$ and $L_{2}$ can benefit from a certain deviation bid $b^{*}$. This contradicts the equilibrium condition and refutes {\bf Case~(b)}.
    \end{itemize}
    In sum, {\bf Case~(II)} gets refuted.
\end{itemize}
Refute our assumption: Bidders $L_{1}$ and $L_{2}$ play truthful strategies $s_{L,\, 1}(v) \equiv s_{L,\, 2}(v) \equiv v$ everywhere $v \in [0,\, 1 - 1 / e]$. This finishes the proof.
\end{proof}

\begin{lemma}
Bidder $H$ has a fixed strategy $s_{H}(v_{H}) \equiv 0$ for $v_{H} \in \{0,\, 1\}$.
\end{lemma}

\begin{proof}
We reuse the arguments for \Cref{lem:BNE_correlated:focal}. Given a value $v_{H} \in \{0,\, 1\}$ and a bid $b \in [0,\, 1 - 1 / e]$, bidder $H$ gains an interim utility $= (v_{H} - b) \cdot V_{L}(b) = \frac{v_{H} - b}{\epsilon + 1 - b} \cdot (\epsilon + 1 / e)$. The zero bid $b = 0$, under the either value $v_{H} \in \{0,\, 1\}$, is always the UNIQUE maximizer for this interim utility formula.
\end{proof}

\section{Bayesian (Coarse) Correlated Equilibria}
\label{sec:BCE_BCCE}
We first introduce the definition of joint strategies and Bayesian (Coarse) Correlated Equilibria. 

\begin{definition}[Joint Strategies]
\label{def:correlated_strategy}
A joint strategy profile $\bs(\bv) \equiv (s_{i}(\bv))_{i \in [n]}$ involves $n \geq 1$ many $n$-variate functions; each one maps the whole value profile $\bv \in \RR^{n}$ to a nonnegative bid $s_{i}(\bv) \geq 0$.
(When the functions $(s_{i}(\bv))_{i \in [n]}$ degenerate into {\em univariates} $s_{i}(\bv) \equiv s_{i}(v_{i})$, the strategy profile $\bs(\bv)$ degenerates into an {\em independent} strategy profile, as before for {\BayesNashEquilibria}.)
\end{definition}

\begin{definition}[{\sf Bayesian} ({\sf Coarse}) {\sf Correlated Equilibria}]
\label{def:correlated_equilibria}
Given a joint value distribution $\bV \in \bbV_{\sf joint}$, a tie-breaking rule $\alloc \in \bbFPA$, and a precision $\delta > 0$:
\begin{itemize}
    \item A $\delta$-approximate {\BayesCorrelatedEquilibrium} $\bs \in \bbBCE(\bV,\, \alloc,\, \delta)$ is a joint strategy profile that, for any bidder $i \in [n]$, value $v \in \supp_{i}(\bv)$, bid $b \in \supp_{i}(\bs(\bv))$, and deviation bid $b^{*} \geq 0$,
    \begin{align*}
        \Ex_{\bv,\, \bs,\, \alloc} \big[\, u_{i}(v_{i},\, \bs(\bv)) \,\bigmid\, v_{i} = v,\, s_{i}(\bv) = b \,\big]
        ~\geq~ \Ex_{\bv,\, \bs,\, \alloc} \big[\, u_{i}(v_{i},\, b^{*},\, \bs_{-i}(\bv)) \,\bigmid\, v_{i} = v,\, s_{i}(\bv) = b \,\big] - \delta.
    \end{align*}
    
    \item A $\delta$-approximate {\BayesCoarseCorrelatedEquilibrium} $\bs \in \bbBCCE(\bV,\, \alloc,\, \delta)$ is a joint strategy profile that, for any bidder $i \in [n]$, value $v \in \supp_{i}(\bv)$, and deviation bid $b^{*} \geq 0$,
    \begin{align*}
        \Ex_{\bv,\, \bs,\, \alloc} \big[\, u_{i}(v_{i},\, \bs(\bv)) \,\bigmid\, v_{i} = v \,\big]
        ~\geq~ \Ex_{\bv,\, \bs,\, \alloc} \big[\, u_{i}(v_{i},\, b^{*},\, \bs_{-i}(\bv)) \,\bigmid\, v_{i} = v \,\big] - \delta.
    \end{align*}
\end{itemize}
It is well-known (see \cite{R15}) that all equilibrium concepts together form the following inclusion: {\BayesNashEquilibrium} $\subseteq$ {\BayesCorrelatedEquilibrium} $\subseteq$ {\BayesCoarseCorrelatedEquilibrium}.
\end{definition}

\begin{theorem}[Tight {\PoS}]
\label{thm:pos_bce}
Given a joint value distribution $\bV \in \bbV_{\sf joint}$ and any tie-breaking rule $\alloc^{*} \in \bbFPA$, for any $\delta > 0$, there exists a joint strategy profile $\bs(\bv) \equiv (s_{i}(\bv))_{i \in [n]}$ such that
\begin{enumerate}[font = {\em\bfseries}]
    \item\label{thm:pos_bce:1} The expected auction/optimal {\SocialWelfares} are equal $\FPA(\bV,\, \alloc^{*},\, \bs) = \OPT(\bV,\, \alloc^{*},\, \bs)$.
    
    \item\label{thm:pos_bce:2} It forms a $\delta$-approximate {\BayesCorrelatedEquilibrium} $\bs \in \bbBCE(\bV,\, \alloc^{*},\, \delta)$ and thus also a $\delta$-approximate {\BayesCoarseCorrelatedEquilibrium} $\bs \in \bbBCCE(\bV,\, \alloc^{*},\, \delta)$.
\end{enumerate}
\end{theorem}

\begin{proof}
Let us consider the first-order valuer $h = h(\bv) \eqdef \argmax(\bv) \in [n]$; breaking ties in favor of the smallest index. We explicitly construct a (deterministic) workable joint strategy profile: \\
(i)~The first-order valuer $h$ bids the second highest value, namely $s_{h}(\bv) = \max(\bv_{-h})$. \\
(ii)~Each other bidder $i \in [n] \setminus \{h\}$ bids her value minus a $\delta$ term, namely $s_{i}(\bv) = v_{i} - \delta$. \\
This strategy profile $\bs$ always allocates the item to the first-order valuer $h$ and realizes the optimal {\SocialWelfare} $= \max(\bv)$. In expectation, we have {\bf \Cref{thm:pos_bce:1}} that $\FPA(\bV,\, \alloc^{*},\, \bs) = \OPT(\bV,\, \alloc^{*},\, \bs)$.

The first-order valuer $h$ realizes a utility $u_{h}(v_{h},\, \bs(\bv)) = v_{h} - s_{h}(\bv) = \max(\bv) - \max(\bv_{-h}) \geq 0$. The threshold winning bid for this bidder $h$ is the highest other bid $\max(\bs_{-h}(\bv)) = \max(\bv_{-h}) - \delta$. Thus with another deviation bid $b_{h}^{*} \geq 0$, bidder $h$ realizes a deviation utility $u_{h}(v_{h},\, b_{h}^{*},\, \bs_{-h}(\bv)) \leq (v_{h} - b_{h}^{*}) \cdot \indicator(b_{h}^{*} \geq \max(\bs_{-h}(\bv))) \leq \max(\bv) - \max(\bv_{-h}) + \delta = u_{h}(v_{h},\, \bs(\bv)) + \delta$, which is at most a $\delta$ increase over the current utility.

Each other bidder $i \in [n] \setminus \{h\}$ realizes a zero utility $u_{i}(v_{i},\, \bs(\bv)) = 0$. The threshold winning bid for this bidder $i$ is the highest other bid $\max(\bs_{-i}(\bv)) = s_{h}(\bv) = \max(\bv_{-h}) \geq v_{i}$. To win in the considered {\FirstPriceAuction} $\alloc^{*} \in \bbFPA$, bidder $i$ must overbid $b_{i}^{*} \geq \max(\bs_{-i}(\bv)) \geq v_{i}$ and realize a nonpositive deviation utility $u_{h}(v_{h},\, b_{h}^{*},\, \bs_{-h}(\bv)) \leq 0$.

In sum, the considered strategy profile $\bs$ forms a $\delta$-approximate {\BayesCorrelatedEquilibrium} $\bs \in \bbBCE(\bV,\, \alloc^{*},\, \delta)$. {\bf \Cref{thm:pos_bce:2}} follows then. This finishes the proof.
\end{proof}

Our $\delta$-approximate {\BayesCorrelatedEquilibrium} has almost the same output as the second-price auction (under the truthful strategies). 
We would like to present \Cref{thm:pos_bce} in terms of a universal $\delta$-approximate {\BayesCorrelatedEquilibrium} for any tie-breaking rule $\alloc^{*} \in \bbFPA$, instead of an exact equilibrium for a particular tie-breaking rule $\alloc$ that is compatible with the underlying value distribution $\bV$. We would avoid the latter, because a typical tie-breaking rule should only depend on the {\em bid profile} and the bidders' identities, while a compatible tie-breaking rule further keeps track of the {\em value profile}.

\section*{Acknowledgements}
We are grateful to Xi Chen for invaluable discussions at the early stage of this work.

Y.J.\ is supported by NSF grants IIS-1838154, CCF-1563155, CCF-1703925, CCF-1814873, CCF-2106429, and CCF-2107187.
P.L.\ is supported by Science and Technology Innovation 2030 – ``New Generation of Artificial Intelligence'' Major Project No.(2018AAA0100903), NSFC grant 61922052 and 61932002, Innovation Program of Shanghai Municipal Education Commission, Program for Innovative Research Team of Shanghai University of Finance and Economics, and Fundamental Research Funds for Central Universities.

\begin{flushleft}
\bibliographystyle{alpha}
\bibliography{main}
\end{flushleft}

\end{document}